\documentclass[aps,pra,twocolumn,superscriptaddress,floatfix,nofootinbib,showpacs,longbibliography,notitlepage]{revtex4-1}

\usepackage[T1]{fontenc}     
\fontencoding{T1} 
\usepackage[utf8]{inputenc}  

\usepackage{tikz}
\usetikzlibrary{calc,arrows}
\input{diagram.sty}
\usepackage{array}

\usepackage{xcolor}
\newcommand{\orcid}[1]{\href{https://orcid.org/#1}{\textcolor[HTML]{A6CE39}{\aiOrcid}}}

\usepackage[british]{babel}  
\usepackage[sc,osf]{mathpazo}
\usepackage[scaled=0.86]{berasans}  
\usepackage[colorlinks=true, citecolor=blue, urlcolor=blue]{hyperref}  
\usepackage{graphicx} 
\usepackage[babel]{microtype}  
\usepackage{amsmath,amssymb,amsthm,bm,amsfonts,mathrsfs,bbm} 

\usepackage{xspace}  
\usepackage{xcolor}
\usepackage{multirow}
\usepackage{array}
\usepackage{bigstrut}
\usepackage{braket}
\usepackage{color}
\usepackage{natbib}
\usepackage{multirow}
\usepackage{mathtools}
\usepackage{float}
\usepackage{comment}
\usepackage[caption = false]{subfig}
\usepackage{subfig}
\usepackage{xcolor,colortbl}
\usepackage{color}
\usepackage{rotating}
\usepackage{enumerate}
\usepackage[shortlabels]{enumitem}

\newcommand{\be}{\begin{equation}}
\newcommand{\ee}{\end{equation}}
\newcommand{\ba}{\begin{eqnarray}}
\newcommand{\ea}{\end{eqnarray}}

\newtheorem{lemma}{Lemma}

\newcommand{\R}{\mathbb{R}}
\newcommand{\C}{\mathbb{C}}
\newcommand{\F}{\mathbb{F}}

\def\>{\rangle}
\def\<{\langle}

\newtheorem{theo}{Theorem}

\newtheorem{prop}{Proposition}
\newtheorem{cor}{Corollary}
\newtheorem{defi}{Definition}

\newtheorem{obs}{Observation}

\usepackage{centernot}
\usepackage{subfig}

\DeclareUnicodeCharacter{202A}{} 
\DeclareUnicodeCharacter{202C}{} 

\definecolor{lime}{HTML}{A6CE39}
\DeclareRobustCommand{\orcidicon}{
	\begin{tikzpicture}
	\draw[lime, fill=lime] (0,0) 
	circle [radius=0.16] 
	node[white] {{\fontfamily{qag}\selectfont \tiny ID}};
	\draw[white, fill=white] (-0.0625,0.095) 
	circle [radius=0.007];
	\end{tikzpicture}
	\hspace{-2mm}
}

\foreach \x in {A, ..., Z}{\expandafter\xdef\csname orcid\x\endcsname{\noexpand\href{https://orcid.org/\csname orcidauthor\x\endcsname}
			{\noexpand\orcidicon}}
}
\begin{document}

\title{Unbounded Quantum Advantage in Communication with Minimal Input Scaling}
\author{Sumit Rout\orcidS{}}
\affiliation{International Centre for Theory of Quantum Technologies (ICTQT), University of Gda{\'n}sk, Jana Ba{\.z}ynskiego 8, 80-309 Gda{\'n}sk, Poland}
\author{Nitica Sakharwade\orcidN{}}
\affiliation{International Centre for Theory of Quantum Technologies (ICTQT), University of Gda{\'n}sk, Jana Ba{\.z}ynskiego 8, 80-309 Gda{\'n}sk, Poland}
\author{‪Some Sankar Bhattacharya‬\orcidB{}}
\affiliation{Física Teòrica: Informació i Fenòmens Quàntics, Universitat Autònoma de Barcelona,
08193 Bellaterra, Spain}
\author{‪Ravishankar Ramanathan\orcidR{}‬}
\affiliation{Department of Computer Science, The University of Hong Kong, Pokfulam Road, Hong Kong}
\author{Pawe{\l} Horodecki\orcidP{}}
\affiliation{International Centre for Theory of Quantum Technologies (ICTQT), University of Gda{\'n}sk, Jana Ba{\.z}ynskiego 8, 80-309 Gda{\'n}sk, Poland}

\begin{abstract}
In communication complexity-like problems, previous studies have shown either an exponential quantum advantage or an unbounded quantum advantage with an exponentially large input set $\Theta(2^{n})$ bit with respect to classical communication $\Theta(n)$ bit. In the former, the quantum and classical separation grows exponentially in input while the latter's quantum communication resource is a constant. Remarkably, it was still open whether an unbounded quantum advantage exists while the inputs do not scale exponentially. Here we answer this question affirmatively using an input size of optimal order. Considering two variants as tasks: 1) distributed computation of relation and 2) {\it relation reconstruction}, we study the one-way zero-error communication complexity of a relation induced by a distributed clique labelling problem for orthogonality graphs. While we prove no quantum advantage in the first task, we show an {\it unbounded quantum advantage} in relation reconstruction without public coins. Specifically, for a class of graphs with order $m$, the quantum complexity is $\Theta(1)$ while the classical complexity is $\Theta(\log_2 m)$. Remarkably, the input size is $\Theta(\log_2 m)$ bit and the order of its scaling with respect to classical communication is {\it minimal}. This is exponentially better compared to previous works. Additionally, we prove a lower bound (linear in the number of maximum cliques) on the amount of classical public coin necessary to overcome the separation in the scenario of restricted communication and connect this to the existence of Orthogonal Arrays. Finally, we highlight some applications of this task to semi-device-independent dimension witnessing as well as to the detection of Mutually Unbiased Bases.
\end{abstract}

\maketitle 

\section{Introduction}
Quantum Shannon theory replaces the classical carrier of information with quantum systems in Shannon's model of communication \cite{wilde2011classical}. This initiated a tide of attempts to understand the advantage of encoding classical information in a quantum system. Over the past few decades, there have been numerous works probing the advantage of quantum resources over classical counterparts in various information-theoretic scenarios. Many of these works provide a deeper insight into quantum theory. Some of these quantum advantages have found practical applications in the field of quantum cryptography \cite{Bennett2014,PhysRevLett.69.1293}, quantum communication \cite{ambainis2008quantum,RevModPhys.82.665,Buhrman2016,BSW10,Tavakoli2020} and quantum computing \cite{Howard2014,Grover1996,doi:10.1098/rspa.1992.0167} to name a few. In a prepare and measure scenario, the major share of effort has been devoted to showing an advantage in quantum communication complexity \cite{Yao,Kushilevitz}. It involves computing the minimum communication required between distant parties to perform a distributed computation of functions \cite{Buhrman98}. 

Karchmer and Wigderson \cite{KW88} initiated the study of the communication complexity of relations and established a connection between the communication complexity of certain types of relations and the complexity of Boolean circuits. In \cite{Raz} Raz provided an example of an exponential gap between the classical and quantum communication complexity for a relation. 

Another closely related line of study has been to explore quantum advantage based on orthogonality graphs-inspired problems. In most cases, orthogonality graphs that yield quantum advantage are not Kochen-Specker colourable (KS-colourable) \cite{Saha2019}, thus connecting this set of tasks to the feature of quantum contextuality \cite{Kochen1975}. 

In terms of communication complexity, it is relevant to look into the tasks that show a large separation between classical and quantum communication resources. In this context, a pertinent question arises regarding the maximum gap between these resources. In this article, we introduce a new task based on the communication complexity of relations called {\it relation reconstruction}. For this task, we identify a class of relations based on graphs, such that there is an unbounded gap between one-way zero-error classical and quantum communication. In \cite{Saha2019, Saha2023}, the authors show the quantum communication complexity advantage for graphs which are not KS colourable. This raises questions regarding the usefulness of KS colourable graphs in demonstrating a similar utility of quantum resources. Our proposed task shows quantum advantage independent of the graphs being KS-colourable (or not). Lastly, the exponential advantage in \cite{Raz} requires an infinite set of inputs. In \cite{Massar01}, the authors show an unbounded separation between quantum and classical resources in the absence of public coin while using an input set which is exponentially large relative to the advantage. This leaves open the possibility of obtaining a similar advantage while using an {\it exponentially smaller input set}. Here we solve this problem affirmatively. Our proposed task requires only a finite set of inputs to establish an unbounded separation between classical and quantum resources. In our communication complexity type task, the input size is of the same order as the advantage. It is practically hard to demonstrate such an unbounded advantage if the input data size is exponentially larger than the advantage (See Section \ref{sec:disc} for detailed discussion). Note that the trade-off here is the fact that, unlike in \cite{Raz,Buhrman98,Massar01}, where the exponential advantages are evident when distant parties share public coins (classical correlations), the unbounded gap presented in this work precludes such possibilities.    

In this work, the relation we considered is specified by the rules of distributed {\it clique labelling problem (CLP)} over a graph. We first study the task of distributed computation of the relation by two parties Alice and Bob, and calculate the one-way zero-error communication complexity of the relation (CCR). For this task where any valid answer belonging to the relation is accepted, we show that there is no advantage in encoding information in quantum systems. However, another version of the task, called relation reconstruction, where Bob's output in different runs should span all correct answers, entails a non-trivial quantum advantage in communication. We refer to the communication cost for this relation reconstruction task as Strong Communication Complexity of Relation (S-CCR). This new task is equivalent to the possibility of reconstructing the relation from the complete observed input-output statistics. Demanding reconstruction is a stronger condition than the communication complexity of relations since a function (a special case of relations) can always be reconstructed from the observed statistics in the latter task, while in general for a relation this does not hold.

\subsection{Main Results}
We consider two distinct scenarios depending on the availability of public coins ({\it i.e.} pre-shared correlation) and direct communication resources between the two parties: (i) the spatially separated parties do not share any public coin and have access to one-way communication, (ii) the communication channels can transmit systems of fixed operational dimensions and the public coin is resourceful. In the first scenario, we find that there exist communication tasks that entail an unbounded separation between the operational dimensions of the classical and quantum message systems. We also demonstrate a quantum advantage for the relation induced by a class of vertex-transitive self-complementary graphs, called {\it Paley graphs}. In the second scenario, we show that there exist communication tasks that imply classical channels are required to be assisted by unbounded amounts of the classical public coin ({\it i.e.} shared randomness) while the quantum channel does not require any additional assistance. Further, we show that there exist graphs for which the task with a classical bit channel requires classical public coin increasing linearly in the number of maximum cliques whereas with quantum public coin ({\it i.e.} quantum correlation) assistance it can be performed by a 1 e-bit-assisted classical bit channel.

\subsection{Outline of the paper}

The article is organised as follows. Section \ref{sec:not} introduces the necessary preliminaries for this work, including orthogonal representations and binary/KS colouring of graphs, which we later apply to our communication tasks. We also discuss the notion of the operational dimension of a system in a theory which would be useful to compare classical and quantum systems as resources for communication. In Sec. \ref{sec:ccr} we first introduce the setup and discuss the communication complexity of relations. We then introduce a variation of the task involving distributed computation of a relation, define the Strong Communication Complexity of Relation and a payoff for this new task. In subsection \ref{subsec:cliquecolouring} and \ref{sec:CLP}, we introduce the notion of clique labelling for a graph and the clique labelling problem which we further elaborately discuss in the later sub-sections. In section \ref{sec:sccind_results} we provide all the results for the two tasks when different resources are accessible to the parties. For the results in subsections \ref{CLP:classical} - \ref{SCLP:other}, we consider one-way noiseless communication as a resource and the parties have access to a local source of randomness or private coins only. On the other hand for the results provided in subsection \ref{subsec:corr}, we consider public coin between parties to be a resource as well when they have access to a noiseless one-way communication channel of bounded dimension. In subsection \ref{CLP:classical} we consider the distributed clique labelling problem, calculate the classical and quantum communication complexity of relation (CCR) specified by this problem, and show that there is no quantum advantage in this task. Then we consider the variation of this task called relation reconstruction where Bob's outputs must span all the correct answers. In \ref{subsec:direct}, we calculate the classical strong communication complexity of the relation (S-CCR) and show that it grows with the order of the graph. Next in \ref{subsubsec:directq} we calculate an upper bound on quantum S-CCR and show that when the orthogonality graphs belong to a certain class then the gap between classical and quantum S-CCR can be arbitrarily large. Subsequently, in subsection \ref{SCLP:other} we show that quantum advantage in S-CCR exists for even a larger class of graphs by explicitly considering Paley graphs as an example. In subsection \ref{subsec:corr}, when an additional classical public coin is allowed between the parties, we calculate the amount of classical public coin assistance to bounded classical communication required when orthogonality graphs belong to some specific class. We show the lower bound on the classical public coin assistance grows with the number of maximum cliques. Additionally, we show the advantage of e-bit assistance to classical channels over assistance from arbitrary amounts of classical public coin.  In Sec. \ref{sec:appl} we list some applications of the proposed communication scenario. Finally in Sec. \ref{sec:disc} we summarise the results and also list some open questions. We also compare the current work with pre-existing results and discuss some foundational insights into the results we have presented.

\section{Preliminaries}\label{sec:not}
In this section, we briefly go over known concepts relevant to the article, including notions of orthogonal representation, binary colouring of graphs which is widely used in the study of contextuality \cite{CSW10} and operational dimension.

\subsection{Graphs, Orthogonal Representation, and Binary Colouring}
A graph $\mathcal{G}=(\mathcal{V},\mathcal{E})$ consists of a set of vertices $\mathcal{V}:=(v_1,v_2,\dots,v_n)$ and a set of edges $\mathcal{E}:=(e_1,e_2,\dots,e_m)$ between the vertices. Additionally, the edges may also have a directional property and a weight, which gives rise to further classifications of directed or undirected graphs and weighted or unweighted graphs. In this work, we consider simple undirected unweighted graphs. A subgraph of a graph $\mathcal{G}$ is a graph $\mathcal{G'}=(\mathcal{V}',\mathcal{E}')$ where $\mathcal{E}'\subseteq \mathcal{E}$ such that $\forall e_i\in \mathcal{E}'$ the vertices connected by $e_i$ belong in $\mathcal{V}'\subseteq \mathcal{V}$. For any graph $\mathcal{G}$, a clique is a fully connected subgraph of $\mathcal{G}$. The size of the clique is given by the number of vertices in the subgraph. A {\it maximum clique} of a graph $\mathcal{G}$ is a clique with the largest size and the number of vertices in it is referred to as {\it clique number}.

Among many different representations of an arbitrary graph, orthogonal representation over complex fields is useful in demonstrating the impossibility of a non-contextual hidden variable model for quantum mechanics \cite{CSW10,Rabelo2014}. Here we make use of a general definition of orthogonal representation. The orthogonal representation of a graph over arbitrary fields is defined as the following \cite{Peeters1996}:

\begin{defi}\label{def:FOR}
Given a graph $\mathcal{G}:=(\mathcal{V},\mathcal{E})$, an orthogonal representation of $G$ over field $\F$ is described by the function $\phi~:\mathcal{V}\rightarrow \F^d$, such that\\
  (i) for any two adjacent vertices $v_i$ and $v_j$, $\<\phi(v_i),\phi(v_j)\>=0$,\\
    (ii) $\phi(v_i)\neq\phi(v_j)$, for all $i\neq j$\\
    where $d$ is the dimension of the vector space over field $\F$ and $\<~ ,~\>$ denotes the scalar product (bilinear form) over field $\F$.\\
    This representation is $\it{faithful}$, if $\<\phi(v_i),\phi(v_j)\>=0$ implies that $v_i$ and $v_j$ are adjacent; and is $\it{orthonormal}$, if $|\phi(v_i)|=1$ for all $v_i\in \mathcal{V}$.
\end{defi}

 In this article, we would only consider the orthogonal representations such that the vectors are normalized, {\it i.e.}, $|\phi(v_i)|=1$ for all $v_i\in \mathcal{V}$. An important problem regarding this representation is to find the minimum dimension, say $d_{\F}$, such that the above definition holds. For such an optimal orthogonal representation we denote the {\it faithful orthogonal range} of the graph $\mathcal{G}$ over field $\F$ as $d_{\F}$ (for example, $d_{\R}$, $d_{\C}$ etc.). A lower and an upper bound to the faithful orthogonal range $d_{\F}$ satisfied over an arbitrary field ${\F}$, are given as follows:
 
 \begin{equation}\label{eq:for}
     \omega \le d_{\F} \le d_{\F'} \le |\mathcal{V}|
 \end{equation}
 
 where, $\F'\subseteq \F$, $\omega$ is the maximum clique size and $|\mathcal{V}|$ is the number of vertices in the graph $\mathcal{G}$ also known as the order of the graph. The lower bound follows from the constraint that there should be at least $\omega$ number of orthogonal vectors for any faithful orthogonal representation. The upper bound says that it is always trivially
possible to provide an orthogonal representation with $|\mathcal{V}|$ number of mutually orthogonal vectors. Lov\'{a}sz $\it{et~al.}$ \cite{lovasz1989orthogonal} provided a necessary and sufficient condition for finding minimal $d$ over the real field $\R$ for a class of orthogonal representations called $\it{general~ position}$, for which any set of $d$ representing real vectors are linearly independent.

\begin{prop}\label{prop:lovasz}[Lov\'{a}sz {\it et al.} '89 \cite{lovasz1989orthogonal}]
Any graph $\mathcal{G}:=(\mathcal{V},\mathcal{E})$, has a general position faithful orthogonal representation in $\mathbb{R}^d$ if and only if at least $(|\mathcal{V}|-d)$ vertices are required to be removed to make the complementary graph $\Bar{\mathcal{G}}$ disconnected.
\end{prop}

In particular, Proposition \ref{prop:lovasz} provides an upper bound on the faithful orthogonal range (definition \ref{def:FOR}) for a class of graphs, a fact that we will use later (subsection \ref{subsubsec:directq}) to arrive at our main result of unbounded quantum advantage in a communication task.
Throughout this work, we will refer to a graph with a faithful orthogonal representation in minimum dimension as an orthogonality graph.

Given a graph $\mathcal{G}$, the problems concerning the colouring of its vertices with one of two possible colours have been widely studied and share deep connections with quantum non-contextuality. In the following, we define the binary colouring of an orthogonality graph.

\begin{defi}\label{def:binary}
Binary colouring of a graph $\mathcal{G}:=(\mathcal{V},\mathcal{E})$ is a binary function $\mathsf{f}~:\mathcal{V}\rightarrow \{0,1\}$ such that\\
    (i) for any two adjacent vertices $v_i$ and $v_j$, $\mathsf{f}(v_i)\mathsf{f}(v_j)=0$,\\
    (ii) for any maximum clique $\mathsf{C}_k$ of the graph $\mathcal{G}$ there is exactly one vertex $v_*\in \mathsf{C}_k$, such that $\mathsf{f}(v_*)=1$.
\end{defi}
A point to note here is that not all graphs are binary colourable. A {\it Binary Colouring} of graph $\mathcal{G}$ with $n$ vertices, if possible, can be thought of as a binary string of length $n$. On the other hand, the set of the binary strings corresponding to all different binary colourings uniquely describes the graph $\mathcal{G}$. In the subsequent sections, we will use the term "colouring of a graph" to refer to the binary colouring of the graph.

\subsection{Comparison of Classical and Quantum Resources}

In any communication protocol, the carrier of the message, as well as the sources of private or public coins, are physical systems, which may be described as classical or quantum (or more generally but outside the purview of this work by a post-quantum theory). To facilitate comparing resources, we will describe below the notion of \textit{Operational dimension} from the framework of General Probabilistic Theory (GPT) \cite{plavala2021}. The {\it operational dimension} of a system is the largest cardinality of the subset of states that a single measurement can perfectly distinguish.

Importantly, the operational dimension of a theory is different from the dimension of the vector space $V$ in which the state space $\Omega$ is embedded. For instance, for qubit the state space, the set of density operators $\mathcal{D}(\mathbb{C}^2)$ acting on $\mathbb{C}^2$ is embedded in $\mathbb{R}^3$. However, the operational dimension of this system is $2$, as at most two-qubit states can be perfectly distinguished by a single measurement. Generally, the operational dimension is equivalent to the Hilbert space dimension for a quantum system. We will refer to this notion when comparing communication resources between the quantum and classical scenarios.

\section{Communication Complexity of Relations}\label{sec:ccr}

In this Section, we will briefly discuss the communication complexity of relations. Consider a bipartite prepare and measure scenario involving Alice and Bob who are separated in space and can communicate. A relation is defined as a subset $\mathcal{R}\subseteq X \times Y \times B$, where $X$ and $Y$ are the set of possible input values of Alice and Bob, respectively, and $B$ is the set of possible output values that can be produced by Bob. A simple example is the relation $\mathcal{R}$ where $X$ and $Y$ are sets of Parents and the set $B$ is the set of Children and a valid tuple $(x,y,b)\in \mathcal{R}$ when $b$ is a child of $x$ and $y$. There might be multiple {\it correct answers} if $x$ and $y$ have multiple children. There is also the possibility of no valid output for a given $x$ and $y$ if they have no children. We will consider relations that have a valid output $b$ for any valid input $(x,y)$. The task for Alice and Bob is distributed computation of relation $\mathcal{R}$. They can use communication as a resource for this purpose. A protocol $\mathsf{P}$ to perform this task may involve one-way or two-way communications with single or multiple rounds. However, in this work, we are interested in {\it one-way} communication protocols only. The cost of a protocol $\mathsf{P}$ is defined as the minimum amount of communication required to perform the distributed computation for any input $(x,y)\in X\times Y$. Now we will define the Communication Complexity of a Relation (or CCR).

\begin{defi} \textbf{CCR}
The communication complexity of a relation $\mathcal{R}\subseteq X \times Y \times B$ is the minimum communication required from Alice to Bob such that for any input variables $x\in X$ and $y\in Y$, Bob's output $b$ gives the tuple $(x,y,b)\in \mathcal{R}$. Note that Alice and Bob should know the relation $\mathcal{R}$ before the task commences.
\end{defi}

 In other words, the communication complexity of the relation $\mathcal{R}$ is the minimum communication required when optimised over all protocols that can compute $\mathcal{R}$. In a generalised setting, the distributed computation task may allow for some small errors to lower the cost of communication. Throughout this article we consider only {\it zero-error } scenario, {\it i.e.} $P(b|x,y)=0$ whenever $(x,y,b)\notin \mathcal{R}$ for all $(x,y)\in X \times Y$. In most cases, rather than finding the optimal protocol or its cost, which is often difficult, one is interested in providing a lower bound for the communication complexity. A trivial {\it zero-error} protocol using $\log_2 |X|$ bit, which requires that Alice sends all information about her input to Bob, provides a trivial upper bound for communication complexity.

 The protocols for the distributed computation of a relation depending on encoding and decoding strategies have the following types. Firstly, the classical one-way $\log_2 m$ bit communication protocol can be deterministic. Such a deterministic protocol consists of a fixed encoding $\mathbb{E}$ by Alice which is a `$\log_2 |X|$ bit to $m$ bit' deterministic function and a decoding by Bob $\mathbb{D}$ which is a `$m\log_2 |Y|$ bit to $\log_2 |B|$ bit' deterministic function, {\it i.e.} $\mathbb{E}:\{1,\cdots,|X|\}\mapsto\{0,\dots,m-1\}$ and $\mathbb{D}:\{0,\dots,m-1\}\times \{1,\cdots,|Y|\}\mapsto\{1,\cdots,|B|\}$. The communication cost of such a protocol is defined as the length of the message in bits sent by Alice on the worst choice of inputs $x$ and $y$. The one-way deterministic zero-error communication complexity of relation $\mathcal{R}$, denoted by $\mathbf{D}(\mathcal{R})$ is the cost of the best protocol ({\it i.e} protocol with minimum communication cost) that allows computation of relation $\mathcal{R}$ without any error. Secondly, the parties can have access to private coins or sources of local randomness. In a classical one-way protocol with private coins, they can locally alternate over the space of all possible encodings ($\mathbb{E}$) and decodings ($\mathbb{D}$) while following some probability distribution $P_{\mathbb{E}}$ and $P_{\mathbb{D}}$ respectively. We denote the private coin-assisted communication complexity of relation $\mathcal{R}$ as $\mathbf{R}_{\it priv}(\mathcal{R})$. Thirdly, the parties can have access to public coins or sources of shared correlations. In a classical one-way protocol with public coins, they can switch between deterministic encoding and decoding schemes following some correlation $P_{\mathbb{E}\times \mathbb{D}}$. Here $P_{\mathbb{E}\times\mathbb{D}}$ is a probability distribution over the space of the Cartesian product of deterministic encodings and decodings. We denote the public coin-assisted communication complexity of relation $\mathcal{R}$ as $\mathbf{R}_{\it pub}(\mathcal{R})$. Note that the sources of shared correlations can be classical or quantum and we will refer to them as classical public coin or quantum public coin respectively. The classical communication complexities of a relation $\mathcal{R}$ satisfy the following ordering:

\begin{align}
   \mathbf{R}_{\it pub}(\mathcal{R})\le \mathbf{R}_{\it priv}(\mathcal{R})\le \mathbf{D}(\mathcal{R})
\end{align}

In the communication complexity of functions, there is only a single {\it correct answer} that Bob may output. The task of communication complexity of relations differs from that of functions since there may be more than one correct answer for Bob. This allows us to define a stronger variation of the distributed computation task that enforces that Bob outputs all correct answers over different rounds of the prepare and measure scenario. We will refer to the minimum communication required for this task as Strong Communication Complexity of Relations (S-CCR). Naturally, when the relation is a function (a subclass of relations) S-CCR and CCR are equal as both the tasks are equivalent in the case of functions. 

\begin{defi} \textbf{S-CCR}
The strong communication complexity of a relation $\mathcal{R}\subseteq X \times Y \times B$ is the minimum communication required from Alice to Bob such that for any input variables $x\in X$ and $y\in Y$, Bob's output $b$ gives the tuple $(x,y,b)$ which belongs to $\mathcal{R}$ and that Bob's output $b$ in different rounds of the prepare and measure scenario spans all valid $b$ for each input $(x,y)$. Same as CCR, Alice and Bob should know the relation $\mathcal{R}$ before the task commences.
\end{defi}

In other words, this new task aims to decipher or {\it reconstruct} the relation $\mathcal{R}$ from the observed statistics $\{(x_i,y_i,b_i)|i=runs\}$. Thus, we will refer to this variant of a distributed computation task as relation reconstruction. In the limit of $runs \rightarrow \infty$ the observed statistics can be used to get the conditional output probability distribution $\{P(b|x,y)\}_{x,y,b}$. Note that for relation reconstruction, the necessary condition to guess or reconstruct $\mathcal{R}$ correctly is given by the non-zero value of the observed conditional probabilities when $(x,y,b)\in \mathcal{R}$ (and zero otherwise) rather than the exact probabilities. However, we can define a natural (but not convex) payoff for relation reconstruction as follows:
\begin{equation}\label{eq:pay}
  \mathcal{P_R}=\min_{(x,y,b)\in \mathcal{R}} P(b|x,y).
  \end{equation}
When optimised over all protocols $\mathsf{P}$ for this task with or without public coins involving them, the best strategy yields the maximum achievable payoff for the given relation which we will refer to as algebraic upper bound $\mathcal{P^*_R}$. This is trivially achieved if Alice communicates her input to Bob and Bob in turn uses this message and his input to give a randomly chosen output from the set of all correct answers in each run.

\begin{equation}
  \mathcal{P^*_R}=\max_{\mathsf{P}} \mathcal{P_R}
  \end{equation}
  
One way to interpret the payoff $\mathcal{P_R}$ is to think of it as related to the probability of success of reconstructing the relation $\mathcal{R}$ (See Appendix \ref{app:rate}). Thus, for the given protocol, the higher the value of $\mathcal{P_R}$, the fewer runs one needs to reconstruct the relation. Note that for the success of reconstruction, we necessarily require $\mathcal{P_R}>0$. It is worth mentioning that we are interested in minimum communication that performs the relation reconstruction task. However, the optimal strategy using the minimum amount of communication may not yield $\mathcal{P^*_R}$. Further, two different sets of resources (communication and/or shared) of the same operational dimension, such as quantum and classical, that individually perform relation reconstruction may also yield different payoffs ($\mathcal{P}^Q_\mathcal{R}$ and $\mathcal{P}^{Cl}_\mathcal{R}$ respectively) when optimised over all the strategies given the type of the resource mentioned above.

In this work, we consider some specific relations induced by orthogonality graphs. These relations are specified by a distributed clique labelling problem. Before introducing the clique labelling problem, let us introduce clique labelling.

\subsection{Binary colouring to clique labelling}\label{subsec:cliquecolouring}
Consider an orthogonality graph $\mathcal{G}$ with the set of vertices $\mathcal{V}$ and maximum clique of size $\omega$. Now additionally consider some indexing of the vertices $\{1,\cdots, |\mathcal{V}|\}$ of the graph. Let us define the ordered (increasing indices) set of vertices belonging to a maximum clique $C_i$ as $\mathcal{V}_{C_i} \subseteq \mathcal{V}$. The binary colour, denoted by $f(.)$ is defined over each vertex (Def. \ref{def:binary}). Now consider the binary colouring of the vertices of the maximum clique $C_i$. Observe that any such binary colour assigns value $1$ to exactly one vertex of this maximum clique, that is

\begin{align}
    \forall~ v\in \mathcal{V}_{C_i},~ f(v)= \delta_{v,v'}~\text{for a $v' \in \mathcal{V}_{C_i} $} 
\end{align}

 There are $\omega$ different binary colourings possible for the maximum clique $C_i$. Now clique labelling can be defined as an invertible map that takes one from a binary colouring of the vertices of some maximum clique to an $\omega$-valued label of this clique.

\begin{defi}\label{def:cliqueclabel}
For some maximum clique $C_i$, clique labelling is a mapping $g_{C_i}$ from the set of binary colourings $\{f(v)\}$ of vertices $\mathcal{V}_{C_i}$ in the clique to a $\omega$-valued label in $ \Omega =\{0,...,\omega-1\}$. For a binary colouring $f(v)$ that assigns colour $1$ to the vertex at $a^{th}$ position index in the ordered set $\mathcal{V}_{C_i}$ the corresponding clique label is $ g_{C_i}= a\in\Omega $.
\end{defi}
 
 The clique label is assigned from $\{0,...,\omega-1\}$ in the following manner. The lowest clique label $0$ is assigned if the vertex with the lowest index in $\mathcal{V}_{C_i}$ is assigned $1$ by the binary colouring, the second lowest clique label $1$ is assigned if the vertex with the second lowest index in $\mathcal{V}_{C_i}$ is assigned $1$ by the binary colouring and so on.

\begin{figure}[h]
    \centering
\begin{DiagramV}[1]{0}{0}
\begin{move}{0,0}
\fill[black] (0 ,0 ) circle (0.25);
\draw (0-1,0) node {\footnotesize  $v_1$};
\fill[black] (2 ,3 ) circle (0.25);
\draw (2,3+1) node {\footnotesize  $v_2$};
\fill[black] (4 ,0 ) circle (0.25);
\draw (4+0.2,0+1) node {\footnotesize  $v_3$};
\fill[black] (7 ,-2.5 ) circle (0.25);
\draw (7,-2.5-1) node {\footnotesize  $v_7$};
\fill[black] (8 ,1 ) circle (0.25);
\draw (8-0.2,1+1) node {\footnotesize  $v_6$};
\fill[black] (10.5 ,1+3 ) circle (0.25);
\draw (10.5,4+1) node {\footnotesize  $v_4$};
\fill[black] (12 ,0.5 ) circle (0.25);
\draw (12+1,0.5) node {\footnotesize  $v_5$};
\draw (2,1.2) node {\footnotesize  $C_1$};
\draw (6.2,-0.6) node {\footnotesize  $C_2$};
\draw (10.2,2) node {\footnotesize  $C_3$};
\draw (0 ,0 ) -- (2 ,3 ) -- (4 ,0 ) -- (0 ,0 );
\draw (8 ,1 ) -- (10.5 ,4 ) -- (12 ,0.5 ) -- (8 ,1 );
\draw (8 ,1 ) -- (4 ,0 ) -- (7 ,-2.5 ) -- (8 ,1 );
\draw (7 ,-5.5 ) node { $\mathcal{G}^{(n=3,\omega=3})$};
\end{move}
\end{DiagramV}
\\
\vspace{0.2cm}
\begin{DiagramV}[1]{0}{0}
\begin{move}{0,0}
\draw (2.2,-0.6) node {\footnotesize  $C_2$};
\fill[black] (0,0) circle (0.25);
\draw (0-2.8,0+0.2) node {\footnotesize $f(v_3)=1$};
\fill[black] (3,-2.5) circle (0.25);
\draw (3,-2.5-1) node {\footnotesize  $f(v_7)=0$};
\fill[black] (4,1) circle (0.25);
\draw (4,1+1) node {\footnotesize  $f(v_6)=0$};
\draw (4,1) -- (0,0) -- (3,-2.5) -- (4,1);
\end{move} 
\end{DiagramV} $\implies g_{C_2}=0$
\\
\begin{DiagramV}[1]{0}{0}
\draw (2.2,-0.6) node {\footnotesize  $C_2$};
\begin{move}{0,0}
\fill[black] (0,0) circle (0.25);
\draw (0-2.8,0+0.2) node {\footnotesize $f(v_3)=0$};
\fill[black] (3,-2.5) circle (0.25);
\draw (3,-2.5-1) node {\footnotesize  $f(v_7)=0$};
\fill[black] (4,1) circle (0.25);
\draw (4,1+1) node {\footnotesize  $f(v_6)=1$};
\draw (4,1) -- (0,0) -- (3,-2.5) -- (4,1);
\end{move} 
\end{DiagramV} $\implies g_{C_2}=1$
\\
\begin{DiagramV}[1]{0}{0}
\draw (2.2,-0.6) node {\footnotesize  $C_2$};
\begin{move}{0,0}
\fill[black] (0,0) circle (0.25);
\draw (0-2.8,0+0.2) node {\footnotesize $f(v_3)=0$};
\fill[black] (3,-2.5) circle (0.25);
\draw (3,-2.5-1) node {\footnotesize  $f(v_7)=1$};
\fill[black] (4,1) circle (0.25);
\draw (4,1+1) node {\footnotesize  $f(v_6)=0$};
\draw (4,1) -- (0,0) -- (3,-2.5) -- (4,1);
\end{move} 
\end{DiagramV} $\implies g_{C_2}=2$
    \caption{
        For the graph $\mathcal{G}^{(n=3,\omega=3)}$ with $n=3$ maximum cliques of size $\omega=3$ given above, the clique labellings of $C_2$ which has vertices $\mathcal{V}_{C_2}=\{v_3,v_6,v_7\}$ for different binary colourings is provided in this figure.}
    \label{fig:cliquelabelling}
\end{figure}

For example, for $\omega=3$ if say maximum clique $C_i$ has vertices $\mathcal{V}_{C_i}=\{v_3,v_6,v_7\}$ and $f(v_3)=1$ then $g_{C_i}=0$, else if $f(v_6)=1$ then $g_{C_i}=1$ and if $f(v_7)=1$ then $g_{C_i}=2$, where $g_{C_i}$ is the $\omega$-valued clique label of maximum clique ${C_i}$. Note that given the index of vertices and a maximum clique, one can always map the clique label back to the binary colouring of the vertices of a maximum clique, {\it i.e.}, the map is invertible. This will be useful for Bob while decoding during the distributed computation of the clique labelling problem.

\subsection{Clique Labelling Problem ($CLP$)}\label{sec:CLP}
Now, we present the class of relations for which we study CCR and S-CCR in this work. These relations are based on the distributed {\it clique labelling problem} over certain graphs. Here we consider graphs along with some faithful orthogonal representation in minimum dimension and refer to them together as an orthogonality graph. Let us now consider an orthogonality graph $\mathcal{G}^{(n,\omega)}$ with $n$ maximum cliques each of size $\omega$ labelled as $C_i$ where $i\in\{1,...,n\}$. We also assume that each vertex belongs to some $\omega$-sized maximum clique. The set of maximum cliques of the graph $\mathcal{G}^{(n,\omega)}$ is denoted as $\mathcal{C}=\{C_1, C_2,\cdots, C_n\}$ and the set of clique labels for each of the maximum cliques is $\Omega=\{0,\cdots,\omega-1\}$. Note that the clique labels are related to the binary colouring of vertices through the definition \ref{def:cliqueclabel}. 

The setup (given in Fig. \ref{fig:setup1}) for the distributed Clique Labelling Problem (CLP) is a prepare-and-measure scenario involving a referee and two spatially separated parties, Alice and Bob. The referee shares the orthogonal graph $\mathcal{G}^{(n,\omega)}$ with some vertex indexing and a faithful orthogonal representation in minimum dimension with Alice and Bob at the beginning. The referee gives Alice the pair $(C_x, a)$ as input: a maximum clique  $C_x$ of size $\omega$ randomly chosen from $\mathcal{G}^{(n,\omega)}$ and a random possible labelling $a$ of the same maximum clique, i.e., $(C_x, a)\in X=\mathcal{C}\times\Omega$. The referee gives a maximum clique $C_y$ of size $\omega$ randomly chosen from $\mathcal{G}^{(n,\omega)}$ to Bob as input, $C_y\in Y=\mathcal{C}$. We will consider the inputs to be uniformly distributed in the sense that $C_x$ and $C_y$ are both randomly chosen from $\mathcal{C}$ and $a$ are uniformly chosen from $\Omega$.

\begin{figure}[h]
    \centering
\begin{DiagramV}[0.7]{0}{0}
\begin{move}{0,0}
\draw (2,1.2) node {\scriptsize  $C_i$};
\draw (6.2,-0.6) node {\scriptsize  $C_j$};
\draw [gray,ultra thin] (10.2,2) node {\scriptsize  $C_k$};
\draw (0 ,0 ) -- (2 ,3 ) -- (4 ,0 ) -- (0 ,0 );
\draw [gray,ultra thin] (8 ,1 ) -- (10.5 ,4 ) -- (12 ,0.5 ) -- (8 ,1 );
\draw  (8 ,1 ) -- (4 ,0 ) -- (7 ,-2.5 ) -- (8 ,1 );
\draw [gray,ultra thin] (0 ,0 ) -- (-3 ,-2) -- (-3.5 ,2 ) -- (0 ,0 );
\draw [gray,ultra thin] (-2,0) node {\scriptsize  $C_l$};
\draw [gray,ultra thin] (2 ,3 ) -- (0.1,5.5) -- (4.1 ,6 ) -- (2,3 );
\draw [gray,ultra thin] (2-0.1,4.7) node {\scriptsize  $C_m$};
\draw [gray,ultra thin] (-6,1) -- (-3.5,2) --(-5,3); 
\draw [gray,ultra thin]  (6,-4.5) -- (7 ,-2.5) --(8.5,-4.2); 
    \fill[gray] (-3 ,-2) circle (0.2);
    \fill[gray] (-3.5 ,2 ) circle (0.2);
\fill[black] (0 ,0 ) circle (0.25);
    \fill[gray] (0.1,5.5) circle (0.2);
    \fill[gray] (4.1 ,6 ) circle (0.2);
\fill[black] (2 ,3 ) circle (0.25);
\fill[black] (4 ,0 ) circle (0.25);
\fill[black] (7 ,-2.5 ) circle (0.25);
\fill[black] (8 ,1 ) circle (0.25);
\fill[gray] (10.5 ,1+3 ) circle (0.25);
\fill[gray] (12 ,0.5 ) circle (0.25);
\end{move} 
\end{DiagramV}
\\
\begin{DiagramV}[0.75]{0}{0}
\begin{move}{0,0}
\draw (3+7.5 , 11) node {$\mathcal{G}^{(n,\omega})$};
\draw [->] (3+7.5 -3, 11+0.2) to [out=180,in=90] (2,8);
\draw [->] (3+7.5 -3, 11-0.2) to [out=180,in=90] (4,8);
\draw [->] (3+7.5 +3, 11) to [out=0,in=90] (18,8);
\fill[blue!80!red!20!,rounded corners=0.1cm] (0,-3) rectangle (6,3);
\draw [rounded corners=0.1cm] (0,-3) rectangle (6,3);
\draw (3,0) node {\large $A$};
\fill[blue!80!red!20!,rounded corners=0.1cm] (0+15,-3) rectangle (6+15,3);
\draw [rounded corners=0.1cm] (0+15,-3) rectangle (6+15,3);
\draw (3+15,0) node {\large $B$};

\draw[->] (6,0) -- (15,0);
\draw[->] (2,6) -- (2,3);
\draw[->] (4,6) -- (4,3);
\draw (2,7) node { $C_x$};
\draw (4,7) node { $a$};
\draw[->] (3+15,6) -- (3+15,3);
\draw[->] (3+15,-3) -- (3+15,-6);
\draw (3+15,7) node { $C_y$};
\draw (3+15,-7) node { $b$};
\draw (3+7.5,+2) node {$d$};
\end{move}
\end{DiagramV}
\caption{In the prepare and measure scenario, Alice's input is a maximum clique from the orthogonal graph $\mathcal{G}^{(n,\omega)}$ and its clique label, {\it i.e.} $(C_x, a)$. Bob's input is some maximum clique $C_y$ in $\mathcal{G}^{(n,\omega)}$. Bob must output a valid clique labelling $b$ for his input maximum clique such that $(C_x,a,C_y,b)$ $\in \mathcal{R}(\mathcal{G}^{(n,\omega)})$. Alice can send a physical system of operational dimension $d$ to Bob.}
    \label{fig:setup1}
\end{figure}

Bob must output a valid labelling $b\in B=\Omega$ for $C_y$ which satisfies the constraints provided below following from the rules of the binary colouring of the orthogonality graph $\mathcal{G}^{(n,\omega)}$. This will subsequently specify the relation that we will consider. We call these constraints as {\it Consistent Labelling of Pairwise Cliques}:
\begin{enumerate}
    \item If Alice and Bob receive two maximum cliques sharing some vertices, the binary colouring of each shared vertex ($0$ or $1$) by Bob should be identical to Alice's colouring of the vertex.
    \item If Alice and Bob receive two different maximum cliques such that a vertex from Alice's input maximum clique is adjacent to some vertex in Bob's input maximum clique, then the vertices belonging to this edge should not both have binary colour $1$.
    \item In all other cases Bob can label the input maximum clique independently of Alice's input label.   
\end{enumerate}

Note that as a consequence of constraint $1$, if Alice and Bob receive the same maximum clique then Bob's clique labelling should be identical to Alice's input clique labelling, i.e. the binary colouring of the vertices of the maximum clique should be the same. Also, there could be graphs where vertices (not in common) from two maximum cliques are adjacent. The constraint $2$ states that whenever these two maximum cliques are inputs of Alice and Bob, each of these adjacent vertices should not be coloured $1$ simultaneously. The conditions for consistent labelling of pairwise cliques are defined w.r.t. binary colourings, which can be then translated to conditions on the input and output clique labelling in $\{0,\dots,\omega-1\}=\Omega$ (definition \ref{def:cliqueclabel}). Alice is allowed to send some communication (either classical or quantum depending on the scenario) to Bob, which we will optimise to find the communication complexity later. We will also consider situations with classical and quantum public coins later in subsection \ref{subsec:corr}.

We will now define the relation $\mathcal{R}(\mathcal{G}^{(n,\omega)})$ that we discuss throughout this work. $\mathcal{R}(\mathcal{G}^{(n,\omega)})\subseteq X\times Y\times B$ is specified by the distributed CLP for the graph $\mathcal{G}^{(n,\omega)}$. Here $X=\mathcal{C}\times\Omega$ and $Y=\mathcal{C}$ are the input sets for Alice and Bob respectively and $B=\Omega$ is the output set of Bob. Thus, $\mathcal{R}(\mathcal{G}^{(n,\omega)})\subseteq (\mathcal{C}\times\Omega)\times\mathcal{C}\times\Omega$ and the tuple $(x,y,b)\equiv((C_x, a), C_y,b) \in \mathcal{R}(\mathcal{G}^{(n,\omega)})$ if the input maximum clique of Alice and Bob, input clique label of Alice and output clique label of Bob satisfies the constraints of consistent labelling of pairwise clique given above for the graph $\mathcal{G}^{(n,\omega)}$. In the subsequent discussions, we will mildly abuse the notion by using $(C_x, a, C_y,b)$ instead of $((C_x, a), C_y,b)$ to denote a tuple belonging to the relation $\mathcal{R}(\mathcal{G}^{(n,\omega)})$.

In the distributed computation of the $\mathcal{R}(\mathcal{G}^{(n,\omega)})$ relation, Bob must output clique label $b$ for his input maximum clique such that the tuple $(x,y,b)\equiv(C_x, a, C_y,b) \in \mathcal{R}(\mathcal{G}^{(n,\omega)})$. Note that having the relation is equivalent to having the graph itself. Additionally, the CCR and S-CCR when considering the relation $\mathcal{R}(\mathcal{G}^{(n,\omega)})$ will be denoted as CC$\mathcal{R}(\mathcal{G}^{(n,\omega)})$ and S-CC$\mathcal{R}(\mathcal{G}^{(n,\omega)})$ respectively.
 
In Sec. \ref{sec:results}, we show that for distributed computation of $\mathcal{R}(\mathcal{G}^{(n,\omega)})$ there exists a protocol such that $\log_2 \omega$ quantum, as well as classical one-way communication from Alice to Bob, accomplish this task. Thus, we do not have any quantum advantage in CCR in this case. However, it is possible to realise unbounded quantum advantage when we look at the classical and quantum S-CCR when considering $\mathcal{R}(\mathcal{G}^{(n,\omega)})$. 

We add one observation here that will become relevant for some of the results in Sec. \ref{sec:results}. For a graph $\mathcal{G}$ to have a faithful orthogonal representation in dimension $\omega$, any two distinct maximum cliques for this graph can have at most $\omega-2$ points in common. Equivalently, every vertex $v$ in $C_i$ that is not in a maximum clique $C_j$ can be orthogonal to at most $\omega-2$ vertices in $C_j$. 

\subsection{Reconstruction of the Relation $\mathcal{R}(\mathcal{G}^{(n,\omega)})$}\label{sec:sccr}

For the distributed computation of $\mathcal{R}(\mathcal{G}^{(n,\omega)})$, Bob must output some label for his input maximum clique such that it follows the constraints enlisted above for the distributed CLP. Let us now consider the stronger version of the distributed computation task -- relation reconstruction, where Bob must span all correct answers. This can be formulated as a task where the inputs and outputs of Alice and Bob are listed at the end of every round. After sufficient iterations, this list is shared with a randomly chosen Reconstructor (Fig. \ref{fig:reconstruction}), who at the beginning does not have any information about the graph and the induced relation thereof. The Reconstructor becomes aware of the cardinality of the input and/or output sets of Alice and Bob from the list. The Reconstructor must reconstruct the relation $\mathcal{R}(\mathcal{G}^{(n,\omega)})$ and thus the graph $\mathcal{G}^{(n,\omega)}$.

\begin{figure}[h]
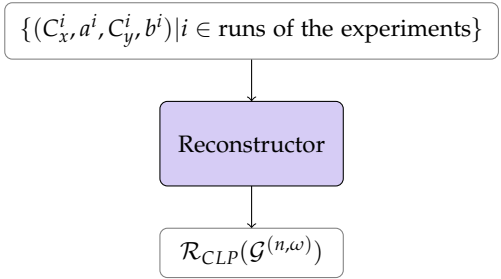

    \centering
\begin{DiagramV}[0.75]{0}{0}
\begin{move}{0,0}
\fill[blue!80!red!20!,rounded corners=0.1cm] (0-1,-3) rectangle (11+1,3);
\draw [rounded corners=0.1cm] (0-1,-3) rectangle (11+1,3);
\draw (5.5,0) node { $\text{Reconstructor}$};
\draw[->] (5.5,6) -- (5.5,3);
\draw[->] (5.5,-3) -- (5.5,-6);
\draw [gray,rounded corners=0.1cm] (-11-1,6) rectangle (11+11+1,10);
\draw (5.5,8) node {$\{(C_x^i,a^i,C_y^i,b^i)|i\in \text{runs of the experiments}\}$};
\draw [gray,rounded corners=0.1cm] (0-1,-6) rectangle (11+1,-9.5);
\draw (5.5,-7.6) node {$\mathcal{R}(\mathcal{G}^{(n,\omega)})$};
\end{move}
\end{DiagramV}
    \caption{\textbf{Reconstruction of Relation} After many runs of the task, the statistics $\{(C^i_x,a^i,C^i_y,b^i)\}$ are sent to the Reconstructor, who attempts to recover $\mathcal{R}(\mathcal{G}^{(n,\omega)})$.}
    \label{fig:reconstruction}
\end{figure}

For reconstruction to be possible, the outcomes of Bob $b$ should be such that after many runs, the set of tuples $\{(C_x,a,C_y,b)\}$ can be used to deduce all the (non-)orthogonality relations in the graph $\mathcal{G}^{(n,\omega)}$ by the Reconstructor, without any prior information about the relation $\mathcal{R}(\mathcal{G}^{(n,\omega)})$ or the graph $\mathcal{G}^{(n,\omega)}$.

After many iterations, the following payoff is calculated for the relation reconstruction task (as defined in \ref{eq:pay}):
\begin{align}
    \mathcal{P}_{\mathcal{R}(\mathcal{G}^{(n,\omega)})}=\min_{(C_x,a,C_y,b)\in \mathcal{R}(\mathcal{G}^{(n,\omega)})} P(b|C_x,C_y,a)\label{eqnpayoff}
\end{align}   
Here the minimisation is over the set of events in $\mathcal{R}(\mathcal{G}^{(n,\omega)})$. The payoff $\mathcal{P}_{\mathcal{R}(\mathcal{G}^{(n,\omega)})}$ is necessarily non-zero if reconstruction is possible. The payoff can be interpreted as a measure of the efficiency or success probability of relation reconstruction over some number of runs. Same as before, when optimised over all protocols with or without public coins, the best strategy yields the maximum achievable payoff for the given relation which we will refer to as algebraic upper bound $\mathcal{P}_{\mathcal{R}(\mathcal{G}^{(n,\omega)})}^*$. 

\subsection{Probability Table for CC$\mathcal{R}(\mathcal{G}^{(n,\omega)})$ and S-CC$\mathcal{R}(\mathcal{G}^{(n,\omega)})$}     

One can analyse the task of distributed computation for the relation $\mathcal{R}(\mathcal{G}^{(n,\omega)})$ as well as the task of relation reconstruction equivalently through a table of conditional probabilities $P(b|C_x,C_y,a)$. This table corresponds to the strategy Alice and Bob can decide on before the game begins while following some protocol. The rows of the table are given by Alice's possible inputs $(C_x,a)$, and the columns are denoted by the tuple of inputs-outputs of Bob $(C_y,b)$. This way of analysis will be important to understand some of the proofs we will present later. The favourable conditions for zero-error distributed computation of $\mathcal{R}(\mathcal{G}^{(n,\omega)})$ (which is {\bfseries (T0)}) and reconstruction of relation $\mathcal{R}(\mathcal{G}^{(n,\omega)})$ (which are {\bfseries (T0)-(T1)}) are provided in terms of the probability table below: 
\begin{enumerate}[start=0,label={(\bfseries T\arabic*):}]
   \item \textbf{Consistent labelling} If a tuple does not belong to the relation then the corresponding conditional probability entry should be zero.
   \begin{align}
       & \forall~ (C_x,a,C_y,b)\notin \mathcal{R}(\mathcal{G}^{(n,\omega)}) \nonumber \\ \implies &P(b|C_x,a,C_y)=0
   \end{align}
   \item \textbf{Relation Reconstruction} If a tuple belongs to the relation, the corresponding conditional probability entry should not be zero.
   \begin{align}
   &\forall ~(C_x,a,C_y,b)\in \mathcal{R}(\mathcal{G}^{(n,\omega)})\nonumber\\ \implies &P(b|C_x,a,C_y)>0
   \end{align}
\end{enumerate}
Further, one can provide an algebraic upper bound $\mathcal{P}_{\mathcal{R}(\mathcal{G}^{(n,\omega)})}^*$ for a given graph $\mathcal{G}^{(n,\omega)}$ from the probability table in the following way. First, fix an input $(\Tilde{C_x},\Tilde{a})$ for Alice and $\Tilde{C_y}$ for Bob. Now count the number of events $(\Tilde{C_x},\Tilde{a},\Tilde{C_y},b)\in \mathcal{R}(\mathcal{G}^{(n,\omega)})$. Lets call this number $\eta (\Tilde{C_x},\Tilde{a},\Tilde{C_y})$. Maximise $\eta (\Tilde{C_x},\Tilde{a},\Tilde{C_y})$ over Alice's and Bob's input sets and call this number $\eta$. Given that the conditions mentioned in {\bfseries (T0)-(T1)} hold, then one has a non-zero payoff for the relation reconstruction task. The payoff satisfies the following inequality
   \begin{align}
       0< \mathcal{P}_{\mathcal{R}(\mathcal{G}^{(n,\omega)}))}\le \frac{1}{\eta}=\mathcal{P}_{\mathcal{R}(\mathcal{G}^{(n,\omega)}))}^*
   \end{align}
     For example, in the case of a graph which has maximum cliques of size $\omega$ that are all disconnected the upper bound on the payoff for the reconstruction of relation  $\mathcal{R}(\mathcal{G}^{(n,\omega)}))$ becomes  $\mathcal{P}_{\mathcal{R}(\mathcal{G}^{(n,\omega)})}\le\frac{1}{\omega}=\mathcal{P}_{\mathcal{R}(\mathcal{G}^{(n,\omega)})}^*$.\\
    
We can now provide the final condition: 

\begin{enumerate}[start=2,label={(\bfseries T\arabic*):}]
\item \textbf{Optimal Payoff} When the payoff $\mathcal{P}_{\mathcal{R}(\mathcal{G}^{(n,\omega)})}$ achieves the algebraic upper bound, we say the payoff is optimal.
   \begin{align}
       0< \mathcal{P}_{\mathcal{R}(\mathcal{G}^{(n,\omega)})}= \frac{1}{\eta}=\mathcal{P}_{\mathcal{R}(\mathcal{G}^{(n,\omega)})}^*
   \end{align}
\end{enumerate}
 
 It is worth highlighting that the payoff $\mathcal{P}_{\mathcal{R}(\mathcal{G}^{(n,\omega)})}$ is a {\it faithful} quantifier of the distributed relation reconstruction problem, that is $\mathcal{P}_{\mathcal{R}(\mathcal{G}^{(n,\omega)})}>0$ whenever relation reconstruction is possible and $\mathcal{P}_{\mathcal{R}(\mathcal{G}^{(n,\omega)})}=0$ implies reconstruction is impossible. Moreover, a protocol that achieves a non-zero payoff for the reconstruction of relation $\mathcal{R}(\mathcal{G}^{(n,\omega)})$ task can trivially perform the distributed computation task for the same relation as well. In some instances, which we will state clearly, our objective will be to additionally maximise the payoff $\mathcal{P}_{\mathcal{R}(\mathcal{G}^{(n,\omega)})}$ for the relation reconstruction task using only specified amount of direct communication resources such as classical bits or qubits and also using shared resources such as classical public coin or quantum public coin (entanglement).

 Appendix \ref{app:example} presents a detailed example to illustrate the clique labelling problem and the reconstruction or relation $\mathcal{R}(\mathcal{G}^{(n,\omega)})$. Additionally, we examine the constraints on the conditional probability table specified earlier for this particular example. In the next section, we present the bulk of our key results while first considering the scenario with only direct communication resources (section \ref{CLP:classical} - \ref{SCLP:other}) and later considering the scenario where the bounded amount of direct communication is assisted by shared resources (i.e. public coins).

\section{One-way Zero-error Classical and Quantum CCR and S-CCR}\label{sec:sccind_results}\label{sec:results}
In the setup described in Section \ref{sec:CLP} (Fig. \ref{fig:setup1}), Alice and Bob have access to a noiseless one-way communication channel of limited capacity  (which is a resource) along with arbitrary local sources of randomness (i.e. private coins) that are considered to be free. Another variation may have them, in addition, sharing some correlations, i.e. public coin.  We will consider the communication scenario when a public coin is not allowed between Alice and Bob for the results discussed in subsections \ref{CLP:classical}-\ref{SCLP:other}. In the subsection \ref{CLP:classical}, we calculate the necessary and sufficient classical and quantum communication (CCR) required to perfectly compute distributed  CLP for any graph $\mathcal{G}^{(n,\omega)}$, and show that there is no advantage offered by quantum theory for this task. Next, we consider the task where Bob's output must span all the correct answers (i.e. relation reconstruction from observed input-output statistics) and calculate the necessary and sufficient classical resource required to accomplish this task (S-CCR) in subsection \ref{subsec:direct}. Subsequently, we calculate the sufficient quantum resource required to accomplish the same task when considering a specific class of orthogonality graphs in subsection \ref{subsubsec:directq}. We show that there is an unbounded separation between quantum and classical resources required to accomplish the task of relation reconstruction for this class of graphs. In subsection \ref{SCLP:other}, we show that there still exists an advantage in using quantum communication resources compared to classical resources for an even larger class of graphs where the orthogonal range is less than the order of the graph, such as the Paley graphs, for which we explicitly show the advantage.  Lastly, in subsection \ref{subsec:corr}, we consider the scenario when public coins are allowed between parties and the dimension of the communication channel is bounded. For certain classes of orthogonality graphs, we show that the necessary amount of classical public coins assistance to a bounded classical communication channel which is required for the relation reconstruction task increases linearly with the number of maximum cliques in the graph. Secondly, we also compare the resourcefulness of quantum public coins with classical public coins when bounded {\it classical} communication is allowed between Alice and Bob to achieve non-zero payoff $\mathcal{P}_{\mathcal{R}(\mathcal{G}^{(n,\omega)})}$.

\subsection{Classical and Quantum Communication Complexity of Relation $\mathcal{R}(\mathcal{G}^{(n,\omega)})$} \label{CLP:classical}

 In this subsection, we calculate the 1) classical and 2) quantum CC$\mathcal{R}(\mathcal{G}^{(n,\omega)})$ some graph $\mathcal{G}^{(n,\omega)}$. The setup is described in section \ref{sec:CLP} and assistance from public coins are forbidden. We will show that quantum and classical CC$\mathcal{R}(\mathcal{G}^{(n,\omega)})$ are the same and thus quantum theory offers no advantage in this communication task.

 We start by observing that both the classical and quantum one-way communication complexity for $\mathcal{R}(\mathcal{G}^{(n,\omega)})$ is bounded from below by the maximum clique size $\omega$ of the given graph. In other words, Alice has to send an $\omega$-level classical (or quantum) system using which Bob can choose a deterministic (or some random correct) output $b$ conditioned on his input $C_y$ and Alice's message. It follows from considering the scenario where both Alice and Bob are given the same maximum clique $C_x=C_y$, Bob must know the input label of Alice (which has the same size as the maximum clique $\omega$) to produce consistent labelling. Further, the quantum protocol can emulate any classical protocol through its coherent version.  Therefore, the objective now reduces to showing that an $\omega$-level classical communication is sufficient for the task.    

\begin{theo}\label{theo:D0classicalcomm}
Given a graph $\mathcal{G}^{(n,\omega)}$, the classical deterministic one way zero error communication complexity of $\mathcal{R}(\mathcal{G}^{(n,\omega)})$ is $\log_2 \omega$ bit.
\end{theo}
The essential idea of the proof is to show that there is a strategy or analogously a table of conditional probabilities $P(b|C_x,C_y,a)$ satisfying {\bfseries (T0)} such that there are $\omega$ distinct rows. Thus, the aforementioned table of conditional probabilities can be compressed into another table with $\omega$ rows only. Alice, upon communicating the row corresponding to her input, enables Bob to output clique labelling consistently depending on his input maximum clique. This strategy (table of conditional probability) involving the communication of $\log_2(\omega)$ bit is necessarily of the following form. For every input of Alice $(C_x,a)$ and Bob $C_y$, there is a deterministic $b$ that Bob chooses to output. This specification is necessary for the probability table to satisfy {\bfseries (T0)}.
\begin{proof} 
For the complete proof see appendix \ref{app:prooftheo1}. 
\end{proof}

Thus, there is no advantage in using quantum resources over its classical analogue when considering the communication complexity of $\mathcal{R}(\mathcal{G}^{(n,\omega)})$.  Note that, the orthogonal representation of the graph which is not relevant to the protocol here will be pertinent in the next two subsections (\ref{subsec:direct}, \ref{subsubsec:directq}), where we will consider the relation reconstruction task and calculate the classical and quantum S-CCR to show an unbounded quantum advantage. 

\subsection{Classical communication cost of {\it relation reconstruction}}\label{subsec:direct}

 In this subsection, we calculate the amount of classical communication necessary and sufficient for the reconstruction of relation $\mathcal{R}(\mathcal{G}^{(n,\omega)})$ from observed input-output statistics when considering the class of orthogonality graphs $\mathcal{G}^{(n,\omega)}$ that satisfy the following conditions:
\begin{enumerate}[start=0,label={(\bfseries G\arabic*):}]
\item each vertex of the graph is part of at least one maximum clique of the graph,
\item $\forall v, v'\in \mathcal{V}$ belonging to two different maximum cliques $\exists~ u\in \mathcal{V}$ such that $u$ is either adjacent to $v$ or $v'$ but not both.
\begin{obs}\label{obs:for}
Given a graph $\mathcal{G}^{(n,\omega)}$ with maximum clique size $\omega$, satisfying conditions ({\bfseries G0})-({\bfseries G1}), there exist induced subgraphs consisting of at least two maximum cliques of size $\omega$, say $C_i$ and $C_j$, such that there is at least one label of $ C_i$ for which there are at least two different consistent choices of labelling for the other maximum clique $C_j$.
\end{obs}
\end{enumerate}

 Given an orthogonality graph $\mathcal{G}^{(n,\omega)}$ satisfying the properties listed above, for the relation $\mathcal{R}(\mathcal{G}^{(n,\omega)})$ we prove a tight lower bound for classical S-CC$\mathcal{R}(\mathcal{G}^{(n,\omega)})$. This bound is calculated for the \textit{zero-error} scenario in which Bob never outputs an outcome $b$ such that the tuple consisting of Alice's and Bob's input, $(C_x,a)$ and $C_y$ respectively, and Bob's output does not belong to the relation $\mathcal{R}(\mathcal{G}^{(n,\omega)})$, i.e., the case $(C_x,a,C_y,b)\notin \mathcal{R}(\mathcal{G}^{(n,\omega)})$ does not occur. 

\begin{lemma}\label{theo:D1classicalcomm}
Given a graph $\mathcal{G}^{(n,\omega)}$ satisfying {(\bfseries G0)}-{(\bfseries G1)}, the zero-error classical strong communication complexity of the relation $\mathcal{R}(\mathcal{G}^{(n,\omega)})$ is 
$\log_2 |\mathcal{V}|$ bit, where $|\mathcal{V}|$ is the order of the graph.
\end{lemma}
\begin{proof}
    See appendix \ref{app:prooflemma1} for the proof.
\end{proof}

 Thus, the zero-error classical S-CC$\mathcal{R}(\mathcal{G}^{(n,\omega)})$ scales linearly with the number of vertices in the graph. In the next section, we calculate sufficient quantum communication that accomplishes the same task when no public coin is allowed between Alice and Bob. We also show, that there exists an unbounded gap between quantum and classical resources when no public coin is pre-shared between the two parties. This separation is observed for a sub-class of graphs considered in this section.

\subsection{Unbounded quantum advantage in {\it relation reconstruction}}\label{subsubsec:directq}
In this subsection, we first calculate the amount of quantum communication sufficient for accomplishing reconstruction of relation $\mathcal{R}(\mathcal{G}^{(n,\omega)})$ when considering the class of orthogonality graphs $\mathcal{G}^{(n,\omega)}$ that also satisfy {(\bfseries G0)}-{(\bfseries G1)}.

\begin{lemma}
 \label{theo:quantumcommunication}
 Given a graph $\mathcal{G}^{(n,\omega)}$ satisfying {(\bfseries G0)}-{(\bfseries G1)} with faithful orthogonal range $d_{\C}$, 
the zero-error quantum strong communication complexity of relation $\mathcal{R}(\mathcal{G}^{(n,\omega)})$ is upper bounded by $\log_2 d_{\C}$ qubits.
\end{lemma}

\begin{proof}
 Alice and Bob are aware of the graph $\mathcal{G}^{(n,\omega)}$ and a faithful orthogonal representation of the graph in dimension $d_{\C}$ before the task. When Alice has access to her input $(C_x,a)$, then she finds the vertex in the maximum clique $C_x$ that is assigned value $1$ by the input clique label $a$. Alice prepares a qudit in the state associated with the orthogonal representation of this vertex and sends the qudit to Bob. Bob then performs a measurement associated with his maximum clique $C_y$. The projectors of the measurement correspond to the orthogonal representation of the vertices in the maximum clique $C_y$. Based on the measurement outcome, which corresponds to some vertex in the maximum clique $C_y$, Bob outputs as his label $b$ that assigns this vertex binary colour $1$. The quantum strategy guarantees consistent labelling of maximum cliques stated equivalently as $\mathcal{R}(\mathcal{G}^{(n,\omega)})$ due to the orthogonal representation. This concludes the proof.
\end{proof}

Now we are in a position to show that there are classes of graphs that give rise to quantum advantage. Lemma \ref{theo:D1classicalcomm}
and Lemma \ref{theo:quantumcommunication} lead us to the following theorem, where we show the condition which guarantees quantum advantage in relation reconstruction task.

\begin{theo}
 \label{theo:qadvantage}
  For any graph $\mathcal{G}^{(n,\omega)}$ satisfying {(\bfseries G0)}-{(\bfseries G1)} with faithful orthogonal range $d_{\C}$, there exists quantum advantage in relation reconstruction while considering $\mathcal{R} (\mathcal{G}^{(n,\omega)})$ whenever $d_{\C}<|\mathcal{V}|$.
\end{theo}
\begin{proof}
The proof of this theorem follows directly from comparing Lemma \ref{theo:D1classicalcomm} and Lemma \ref{theo:quantumcommunication}. 
\end{proof}

The problem of finding the smallest dimension in which a given graph $\mathcal{G}^{(n,\omega)}$ has a (faithful) orthogonal representation is known to be quite difficult \cite{lovasz1989orthogonal, maehara1990dimension}. Since the existence of quantum advantage relies on the faithful orthogonal range being smaller than the order of the graph (theorem \ref{theo:qadvantage}), it follows that the problem of defining the set of graphs that entail quantum advantage is at least as complex as providing a non-trivial upper bound to the faithful orthogonal range for any arbitrary graph. Despite this difficulty, one can identify some families of graphs that are useful for demonstrating an unbounded separation between classical and quantum S-CC$\mathcal{R}(\mathcal{G}^{(n,\omega)})$. Let us consider graphs $\mathcal{G}^{(n,\omega)}$ that satisfy the following condition along with {\bfseries (G0) - (G1)}:
\begin{enumerate}
\item[(\bfseries G2):] At least $(V-k)$ vertices are required to be removed from the graph, where $k\in \mathbb{N}$ and $\omega\leq k< |\mathcal{V}|$, such that the complementary graph is fully disconnected.
\end{enumerate}
The following are the example of such graphs $\mathcal{G}^{(n,\omega)}$ satisfying the conditions {\bfseries (G0) - (G2)}:
\begin{itemize}
    \item[(1)] {\bf Disconnected graphs $\mathcal{G}_{disc.}^{(n,\omega)}$:} Graphs with $n$ maximum cliques of size $\omega$ all of which are disconnected from one another. Thus, we have $|\mathcal{V}|=n\omega$. See Fig. \ref{fig:examples} (a) for an example.
    \item[(2)] {\bf Nearest Neighbour Connected Cliques $\mathcal{G}_{NNCC(r)}^{(n,\omega)}$:} Graph with a chain of $n$ maximum cliques of size $\omega$ such that only maximum clique $C_i$ and $C_{i+1}$ share $r$ ($1\leq r<\frac{\omega}{2}$) vertices where $i\in\{1,2,\cdots,n-1\}$. The rest of the maximum cliques do not share any additional vertices and edges. See Fig. \ref{fig:examples} (b) for an example.
    \item[(3)] {\bf Paley graphs $\mathcal{G}_{Paley(q)}$:} The class of Paley graphs. (See subsection \ref{SCLP:other})
\end{itemize}

\begin{figure}[h]
        \centering
    \begin{DiagramV}[0.8]{0}{0}
    \begin{move}{0,0} 
    \fill[black] (0,0) circle (0.25);
    \fill[black] (6,0) circle (0.25);
    \fill[black] (3,5.196) circle (0.25);
    \draw (0,0) -- (6,0) --(3,5.196)--(0,0);
    \draw (3,5.196/2 -0.5) node{\tiny $C_1$};
    \fill[black] (3+7,5.196) circle (0.25);
    \fill[black] (6+7,0) circle (0.25);
    \fill[black] (+7,0) circle (0.25);
    \draw (0+7,0) -- (6+7,0) --(3+7,5.196) --(0+7,0);
    \draw (3+7,5.196/2 -0.5) node {\tiny $C_2$};
    %
    \fill[gray] (6+14,0) circle (0.25);
    \fill[gray] (6+8,0) circle (0.25);
    \fill[gray] (3+14,5.196) circle (0.25);
    \draw [gray,ultra thin] (0+14,0) -- (6+14,0) --(3+14,5.196) --(0+14,0);
    \draw [gray,ultra thin] (3+14,5.196/2 -0.5) node { \tiny$C_3$};
    \fill[gray] (6+16,0) circle (0.25);
    \fill[gray] (6+18,0) circle (0.25);
    \fill[gray] (6+14,0) circle (0.25);
    \fill[gray] (6+24,0) circle (0.25);
    \fill[gray] (3+24,5.196) circle (0.25);
    \draw [gray,ultra thin] (0+24,0) -- (6+24,0) --(3+24,5.196) --(0+24,0);
    \draw [gray,ultra thin] (3+24,5.196/2 -0.5) node {\tiny $C_{n-1}$};
    \fill[black] (0+31,0) circle (0.25);
    \fill[black] (6+31,0) circle (0.25);
    \fill[black] (3+31,5.196) circle (0.25);
    \draw (0+31,0) -- (6+31,0) --(3+31,5.196) --(0+31,0);
    \draw (3+31,5.196/2 -0.5) node {\tiny $C_n$};
    \end{move}
    \end{DiagramV}    \\
    \text{(a) Disconnected Graphs for $\omega=3$}\\
     \begin{DiagramV}[0.8]{0}{0}
    \begin{move}{0,0} 
    \fill[black] (0,0) circle (0.25);
    \fill[black] (6,0) circle (0.25);
    \fill[black] (3,5.196) circle (0.25);
    \draw (0,0) -- (6,0) --(3,5.196)--(0,0);
    \draw (3,5.196/2 -0.5) node{\tiny $C_1$};
    \fill[black] (3+6,5.196) circle (0.25);
    \fill[black] (6+6,0) circle (0.25);
    \draw (0+6,0) -- (6+6,0) --(3+6,5.196) --(0+6,0);
    \draw (3+6,5.196/2 -0.5) node {\tiny $C_2$};
    %
    \fill[gray] (6+12,0) circle (0.25);
    \fill[gray] (3+12,5.196) circle (0.25);
    \draw [gray,ultra thin] (0+12,0) -- (6+12,0) --(3+12,5.196) --(0+12,0);
    \draw [gray,ultra thin] (3+12,5.196/2 -0.5) node { \tiny$C_3$};
    \fill[gray] (6+16,0) circle (0.25);
    \fill[gray] (6+18,0) circle (0.25);
    \fill[gray] (6+14,0) circle (0.25);
    \fill[gray] (6+24,0) circle (0.25);
    \fill[gray] (3+24,5.196) circle (0.25);
    \draw [gray,ultra thin] (0+24,0) -- (6+24,0) --(3+24,5.196) --(0+24,0);
    \draw [gray,ultra thin] (3+24,5.196/2 -0.5) node {\tiny $C_{n-1}$};
    \fill[black] (0+30,0) circle (0.25);
    \fill[black] (6+30,0) circle (0.25);
    \fill[black] (3+30,5.196) circle (0.25);
    \draw (0+30,0) -- (6+30,0) --(3+30,5.196) --(0+30,0);
    \draw (3+30,5.196/2 -0.5) node {\tiny $C_n$};
    \end{move}
    \end{DiagramV}  
    \text{(b) Graphs with Nearest Neighbour Connected Cliques}\\
    \text{for $\omega=3$ and $r=1$}
    \caption{An example for $\mathcal{G}^{(n,\omega=3)}$ of disconnected graphs (top (a) ) and graphs with nearest neighbour connected cliques (bottom (b) ).}
    \label{fig:examples}
    \end{figure}

 Having shown that for the relation $\mathcal{R}(\mathcal{G}^{(n,\omega)})$ there is a class of graphs showing a quantum advantage in relation reconstruction, we now address the question regarding the extent to which the separation between these resources can be extended.

\begin{theo}
\label{theo:unbounded}
 For the class of graphs $\mathcal{G}^{(n,\omega)}$ satisfying conditions {(\bfseries G0)}-{(\bfseries G2)} with faithful orthogonal range $\omega$, the separation between one-way zero-error classical and quantum S-CC$\mathcal{R}(\mathcal{G}^{(n,\omega)})$ is unbounded. 
\end{theo}

\begin{proof}
Consider the class of graphs $\mathcal{G}^{(n,\omega)}$ satisfying conditions {(\bfseries G0)}-{(\bfseries G2)} with $k=\omega$. 
For instance, a graph $\mathcal{G}^{(n,\omega)}$ with a chain of $n$ maximum cliques of size $\omega$ such that only maximum clique $C_i$ and $C_{i+1}$ share $r$ ($0\leq r<\frac{\omega}{2}$) vertices where $i\in\{1,2,\cdots,n-1\}$. The rest of the maximum cliques do not share any additional vertices or edges than defined above. Thus, the number of vertices of the graph is $|\mathcal{V}|=n(\omega - r) + r$.  

 From lemma \ref{theo:D1classicalcomm}, the zero-error classical S-CC$\mathcal{R}(\mathcal{G}^{(n,\omega)})$ is $\lceil \log_2 \{n(\omega - r) + r\}\rceil$ bit. On the other hand, Lemma \ref{theo:quantumcommunication} implies that protocols using quantum resources can achieve the same by communicating $\lceil \log_2 \omega\rceil$ qubits, provided the graph $\mathcal{G}^{(n,\omega)}$ has a faithful orthogonal range $d_{\C}=\omega$. 

According to Lov\'{a}sz's theorem (see Section \ref{sec:not}, Proposition \ref{prop:lovasz}) \cite{lovasz1989orthogonal} a faithful orthogonal representation of the graph $\mathcal{G}^{(n,\omega)}$ exists in dimension $d_{\R}=\omega$, since it is necessary to remove at least $(n\omega-\omega)$ vertices from the complementary graph $\Bar{\mathcal{G}}^{(n,\omega)}$ to make it completely disconnected. It also follows from Eq. (\ref{eq:for}) that for the graph $\mathcal{G}^{(n,\omega)}$, the faithful orthogonal range over complex field $d_{\C}=\omega$. As one can obtain such a faithful orthogonal representation of the graph $\mathcal{G}^{(n,\omega)}$ in dimension $d_{\C}=\omega$ and therefore the separation between classical ($\lceil \log_2 \{n(\omega - r) + r\}\rceil$ bit) and quantum ($\lceil \log_2 \omega\rceil$ qubits) communication can be made unbounded by considering large $n$.   
\end{proof}

Given any graph $\mathcal{G}^{(n,\omega)}$, having an orthogonal range $d_{\C}=\omega$ and satisfying conditions ({\bfseries G0})-({\bfseries G1}), the maximum payoff $\mathcal{P}_{\mathcal{R}(\mathcal{G}^{(n,\omega)})}$  achievable for relation reconstruction by direct quantum communication resource of operational dimension $\omega$ is connected to the optimal faithful orthogonal representation of the graph within dimension $d_{\C}$. To see this, notice that the maximum payoff for the quantum strategy is given by the maximisation of the minimum overlap of the vectors corresponding to any two disconnected vertices of the graph (following the same protocol as in Lemma \ref{theo:quantumcommunication}). 
So, keeping in mind the correspondence between quantum strategy and faithful orthogonal representation of the graph $\mathcal{G}^{(n,\omega)}$, one can rephrase the payoff (Eqn \ref{eqnpayoff}) with communication of $d=\omega$-dimensional quantum system, as an optimisation over the faithful orthogonal representations of the graph $\mathcal{G}^{(n,\omega)}$ in dimension $\omega$ on the complex field,
{\it i.e.}
\begin{align}\label{Eq:theta}
 \mathcal{P}_{\mathcal{R}(\mathcal{G}^{(n,\omega)})}^{\C^{\omega}_{max}}&=\min_{(C_x,a,C_y,b)\in \mathcal{R}(\mathcal{G}^{(n,\omega)})} P(b|C_x,C_y,a)\\
 &=\max_{FOR(\C^{\omega})}\Bigl\{\min_{(C_x,a,C_y,b)\in \mathcal{R}(\mathcal{G}^{(n,\omega)})} Tr[\Pi^{C_x}_a\Pi^{C_y}_b]\Bigr\}\\
 &=\max_{FOR(\C^{\omega})}\min_{(i,j)\notin \mathcal{E}} |\langle v(i), v(j)\rangle|^{2}  
\end{align}
where, $FOR(\C^{\omega})$ denotes the set of all faithful orthogonal representations  in dimension $\omega$ over complex field. This relation connects a property of the graph $\mathcal{G}^{(n,\omega)}$ (on the right) to an operational quantity (on the left).

\subsection{Quantum advantage in {\it relation reconstruction} for other graphs} \label{SCLP:other}
In this section, we will consider a particular class of orthogonality graphs $(\mathcal{G}^{(n,\omega)},\mathcal{V},\mathcal{E})$ called Paley graphs. This class of graphs have been well studied in graph theory \cite{elsawy2012paley} and has found applications in quantum information \cite{Naghipour2015,Gravier2013}. They satisfy the properties {\bfseries (G0)-(G1)} (see observation \ref{obs:paleyg1}).  Note that we already know that for graphs satisfying ({\bfseries G0})-({\bfseries G1}), the classical strong communication complexity increases with the order of the graph, {\it i.e.} $\log_2 |\mathcal{V}|$ bit (Lemma \ref{theo:D1classicalcomm}). Thus, graphs with orthogonal range strictly less than its order entails a quantum advantage in communication (following the same protocol described in the proof of Lemma \ref{theo:quantumcommunication}) when considering the one-way strong communication complexity of relation $\mathcal{R}(\mathcal{G}^{(n,\omega)})$. For the class of well-known Paley graphs, we will show that it has a faithful orthogonal representation in a dimension slightly more than half of the order of the graph (see Theorem \ref{theo:paleyg2}).

\subsubsection{Paley graphs}
 Paley graphs $\mathcal{G}_{Paley(q)}$ are simple undirected graphs whose vertices denote the elements of a finite field $\mathbb{F}_q$ (of order prime power $q  = 4k+1$ for positive integer $k$), and whose edges denote that the corresponding elements differ by a quadratic residue. Paley Graphs have the interesting property that they are vertex-transitive, self-complementary graphs which means that by Lov\'{a}sz's original result, the value of $\theta(\mathcal{G}_{Paley(q)})$ can be computed exactly to be $\theta(\mathcal{G}_{Paley(q)}) = |V(\mathcal{G}_{Paley(q)})|^{1/2} = \sqrt{q}$. 
 Some simple Paley graphs are shown in figure \ref{fig:paleygraph}. Next, we will show that the class of Paley graphs satisfy the condition {\bfseries (G1)}.

\begin{figure}[h]
    \centering
\begin{DiagramV}[1.5]{0}{0}
\begin{move}{0,0}
\fill[black] (0,1.5*2.01+0.6) circle (0.25);
\fill[black] (1.5*1.91,1.5*0.62+0.6) circle (0.25);
\fill[black] (1.5*1.18,1.5*-1.63+0.6) circle (0.25);
\fill[black] (1.5*.82-1.5*2,1.5*-1.63+0.6) circle (0.25);
\fill[black] (1.5*0.09-1.5*2,1.5*0.62+0.6) circle (0.25);
\draw (0,1.5*2.01+0.6) -- (1.5*1.91,1.5*0.62+0.6)--(1.5*1.18,1.5*-1.63+0.6) -- (1.5*.82-1.5*2,1.5*-1.63+0.6) -- (1.5*0.09-1.5*2,1.5*0.62+0.6) --(0,1.5*2.01+0.6) ;
\draw (0,-4) node {\large $\mathcal{G}_{Paley(\mathbb{F}_5)}$};
\end{move}
\end{DiagramV}
~~~
\begin{DiagramV}[1.5]{0}{0}
\begin{move}{0,0}
\fill[black] (-3,-1) circle (0.25);
\fill[black] (3,-1) circle (0.25);
\fill[black] (0,3*1.73-1) circle (0.25);
\fill[black] (-2.5,1.9*1.73-1) circle (0.25);
\fill[black] (2.5,1.9*1.73-1) circle (0.25);
\fill[black] (0,-0.7-1) circle (0.25);
\fill[black] (-1,1.4*1.73-1) circle (0.25);
\fill[black] (1,1.4*1.73-1) circle (0.25);
\fill[black] (0,+0.7-1) circle (0.25);
\draw (-3,-1) -- (0,0.7-1) -- (3,-1) -- (1,1.4*1.73-1)-- (0,3*1.73-1) --(-1,1.4*1.73-1) -- (-3,-1);
\draw (-3,-1) -- (0,-0.7-1) -- (3,-1) -- (2.5,1.9*1.73-1)-- (0,3*1.73-1) --(-2.5,1.9*1.73-1) -- (-3,-1);
\draw (1,1.4*1.73-1) -- (0,-0.7-1) -- (-1,1.4*1.73-1) -- (2.5,1.9*1.73-1) -- (0,0.7-1) --(-2.5,1.9*1.73-1) --(1,1.4*1.73-1);
\draw (0,-4) node {\large $\mathcal{G}_{Paley(\mathbb{F}_9)}$};
\end{move}
\end{DiagramV}
\caption{Example of the 5-Paley graph $\mathcal{G}_{Paley(\mathbb{F}_5)}$ (left) and the 9-Paley graph $\mathcal{G}_{Paley(\mathbb{F}_9)}$ (right).}
    \label{fig:paleygraph}
\end{figure}
\begin{obs}
\label{obs:paleyg1}
    In the class of Paley graphs, any two vertices in the graph have the same degree, i.e. in a graph with $q$ vertices, each vertex has $\frac{q-1}{2}$ neighbours. Every two adjacent vertices have $\frac{q-5}{4}$ common neighbours and every two non-adjacent vertices have $\frac{q-1}{4}$ common neighbours \cite{elsawy2012paley}. Thus, for every pair of different vertices $v,v'$ there exists a third vertex $u$ that is adjacent to exactly one of the vertices $v$ or $v'$. This implies that condition ({\bfseries G1}) is satisfied by Paley graphs.
\end{obs}

\subsubsection{Quantum advantage in 
S-CCR for Paley graphs}
 We will show that there exists a $FOR$ for Paley graph $\mathcal{G}_{Paley(q)}$ in dimension $\frac{q+1}{2}$ where $q$ is the order of the graph. Further, we show that
the quantum protocol achieves the maximum payoff $\frac{2}{\sqrt{q}+1}$ when following the protocol mentioned in Lemma \ref{theo:quantumcommunication}.\\

We note that $\theta(\mathcal{G}_{Paley(q)})$ can be computed using the semi-definite programming formulation given as
\begin{eqnarray*}
\label{eq:theta-Paley}
\theta(\mathcal{G}_{Paley(q)}) = \max_{M = (M_{i,j})_{i,j=1}^q} \sum_{i,j = 1}^q M_{i,j} \; 
\end{eqnarray*}
\begin{equation}
    \text{s.t.} \; M \succeq 0, \; \sum_{i} M_{i,i} = 1.
\end{equation}

Let $\Gamma_{Paley(q)}$ denote the automorphism group of $\mathcal{G}_{Paley(q)}$, i.e., the set of all permutations $\sigma$ that preserve the adjacency structure of the graph.
Suppose $M$ is an optimal solution point for the optimisation in \eqref{eq:theta-Paley}, then $M^* = \frac{1}{|\Gamma_{Paley(q)}|} \sum_{\sigma \in \Gamma_{Paley(q)}} \sigma^T M \sigma$ also satisfies the constraints of positive semi-definiteness, trace one and the sum over entries being equal to $\theta(\mathcal{G}_{Paley(q)})$. Since $\mathcal{G}_{Paley(q)}$ is vertex-transitive, the sum over permutations in $\Gamma_{Paley(q)}$ goes over transpositions between every pair of vertices so that $M^*_{i,i} = 1/q$ for all $i \in [q]$. $M^*$ is the Gram Matrix of a set of vectors (each of norm $1/\sqrt{q}$) forming an orthogonal representation of $\mathcal{G}_{Paley(q)}$. Let us denote by $S_{opt} = \big\{|u_1 \rangle, \ldots, |u_q \rangle \big\}$ the corresponding set of normalised vectors forming the optimal solution to the Lov\'{a}sz-theta optimisation, and by $M_{opt} = q M^*$ the corresponding Gram Matrix. We see that
\begin{equation}
\theta(\mathcal{G}_{Paley(q)}) = \sum_{i,j = 1}^q \frac{1}{q} \langle u_i | u_j \rangle.
\end{equation}
In other words, we have $\sum_{i,j=1}^q \langle u_i | u_j \rangle = q^{3/2}$. By symmetry and the fact that every vertex in $\mathcal{G}_{Paley(q)}$ has degree $(q-1)/2$ it also follows that $\langle u_i| u_j \rangle = (q^{3/2} - q)/(q(q-1)/2) =  2/(q^{1/2}+1)$ for $i \nsim j$.

Let us now compute the dimensionality of the vectors $|u_i \rangle$ in $S_{opt}$ that form the optimal representation giving rise to $\theta(\mathcal{G}_{Paley(q)})$. This quantity is the dimension of the vectors giving rise to the faithful representation $S_{opt}$ that is traditionally denoted as $\xi^*(\mathcal{G}_{Paley(q)})$. 
\begin{theo}
\label{theo:paleyg2}
The dimension of the optimal representation of $\mathcal{G}_{Paley(q)}$ that gives rise to $\theta(\mathcal{G}_{Paley(q)})$ is $(q+1)/2$.
\end{theo}

\begin{proof}
    See appendix \ref{app:prooftheo4} for the proof.
\end{proof}

From equation \ref{eqnpayoff}, when using the protocol mentioned in Lemma \ref{theo:quantumcommunication}, the payoff function defined in for a graph $\mathcal{G}$ assumes the form shown below:
\begin{equation}
\label{eq:payoff}
\mathcal{P}_{\mathcal{R}(\mathcal{G})}^{{\C}^d}
= \max_{FOR_d(\mathcal{G})} \min_{(i,j) \notin E(\mathcal{G})} |\langle v_i | v_j \rangle|^2,
\end{equation}
where $FOR_d(\mathcal{G})$ denotes the set of faithful orthogonal representations in dimension $d$ for $\mathcal{G}$. Let us compute this function for the class of Paley graphs. 
Firstly, we consider 
\begin{equation}
\label{eq:payoff-2}
\mathcal{P}_{\mathcal{R}(\mathcal{G}_{Paley(q)})}^{{\C}^d}
\leq \max_{FOR(\mathcal{G}_{Paley(q)})} \min_{(i,j) \notin E(\mathcal{G}_{Paley(q)})} |\langle v(i) | v(j) \rangle|^2,
\end{equation}
where $FOR(\mathcal{G}_{Paley(q)})$ denotes the set of faithful orthogonal representations of $\mathcal{G}_{Paley(q)}$ in any dimension.

For $(k,l) \notin E(\mathcal{G}_{Paley(q)})$, let $S^{(k,l)}$ denote a point in $FOR(\mathcal{G}_{Paley(q)})$ that achieves the maximum for the optimisation problem in \eqref{eq:payoff-2} with the minimum being realised at $(k,l) \notin E(\mathcal{G})$. That is, $S^{(k,l)} = \big\{|v_1^{(k,l)} \rangle, \ldots, v_q^{(k,l)} \rangle \big\}$ with $\langle v_i^{(k,l)} | v_j^{(k,l)} \rangle = 0$ for $(i,j) \in E(\mathcal{G}_{Paley(q)})$ and $|\langle v_k^{(k,l)} | v_l^{(k,l)} \rangle|^2 \leq |\langle v_{k'}^{(k,l)} | v_{l'}^{(k,l)} \rangle|^2$ for any $(k',l') \in E(\mathcal{\overline{G}}_{Paley(q)})$, $(k',l') \neq (k,l)$. We claim that $S^{(k,l)} = S_{opt}$, that is, the set of vectors realising the optimal value in the Lov\'{a}sz-theta optimisation. To this end, we claim that 
\begin{equation}
|\langle v_k^{(k,l)} | v_l^{(k,l)} \rangle| \leq \frac{2}{\sqrt{q} + 1}.
\end{equation}
For suppose that $|\langle v_k^{(k,l)} | v_l^{(k,l)} \rangle| > \frac{2}{\sqrt{q} + 1}$. Then consider the Gram Matrix $M^{(k,l)}$ formed by the set of normalised vectors in $S^{(k,l)}$. We see that $(1/q) M^{(k,l)}$ also satisfies the constraints of positive semi-definiteness and trace one for the Lov\'{a}sz-theta optimisation in Eq.(\ref{eq:theta-Paley}). But if the minimum non-zero off-diagonal entry of $(1/q)M^{(k,l)}$ is larger than the minimum non-zero off-diagonal entry of the optimal matrix $M^*$ (with both matrices having diagonal entries all equal to $(1/q)$) then we obtain that $\sum_{i,j=1}^q (1/q) \left(M^{(k,l)} \right)_{i,j} > \sum_{i,j=1}^q \left(M^* \right)_{i,j} = \theta(\mathcal{G}_{Paley(q)})$ which is a contradiction. Therefore, we must have that the quantum maximum value of the payoff function is at most
\begin{equation}
|\langle v_k^{(k,l)} | v_l^{(k,l)} \rangle|^2 = \left( \frac{2}{\sqrt{q} + 1} \right)^2,
\end{equation}
with the maximum achieved by the set of vectors $S_{opt}$ in $\mathbb{R}^{(q+1)/2}$ that also incidentally achieve the optimum value of Lov\'{a}sz-theta for the graph $\mathcal{G}_{Paley(q)}$. 

\subsection{{\it Relation reconstruction} with Public coins}\label{subsec:corr}
In the previous subsections, we considered the strong communication complexity of relation when public coins between Alice and Bob were not allowed. Here, we consider that the parties have access to public coins along with one-way direct communication resources. In public coin-assisted communication complexity problems, usually, the amount of communication necessary and/or sufficient is studied. For this purpose, an unbounded amount of public coin is allowed to be shared between the parties.  However, here we allow for restricted direct communication, either quantum or classical, and compare the amount of classical public coin / shared randomness assistance required for the relation reconstruction from the observed input-output statistics when considering $\mathcal{R}(\mathcal{G}^{(n,\omega)})$. We show that there exist graphs for which a non-zero payoff while using restricted classical communication implies the presence of a public coin.\\

For a class of graphs $\mathcal{G}^{(n,\omega)}$ satisfying {\bfseries (G0)-(G2)} and with faithful orthogonal range $\omega$, we provide a lower bound on the amount of classical public coin required for accomplishing the relation reconstruction when communicating $\log_2 \omega$ bit. We show that this lower bound grows as $\log_2 n$ with the number of maximum cliques $n$. Later on, we also show the lower bound on the amount of public coin which is necessary to achieve optimal payoff $\mathcal{P}_{\mathcal{R}(\mathcal{G}^{(n,\omega}))}^{*}$  for relation reconstruction is connected to the existence of Orthogonal Arrays (OA).  On another note, we then show that there are graphs for which both quantum and classical communication using a $\omega$-dimensional system require the assistance of public coins to achieve optimal payoff for the reconstruction of relation $\mathcal{R}(\mathcal{G}^{(n,\omega)})$. In the end, we also compare the amount of quantum and classical public coin that is required when only a restricted amount of one-way classical communication is allowed to perform relation reconstruction for some specific graphs. In these cases, we show there is an unbounded gap between the amount of quantum and classical public coin.

\subsubsection{Classical communication assisted by classical public coin}\label{subsubsec:corrc}

In Theorem \ref{theo:D0classicalcomm}, we showed that the communication complexity of $\mathcal{R}(\mathcal{G}^{(n,\omega)})$ is $\log_2\omega$ bit. It is the minimum communication required for satisfying ({\bfseries T0}). Then in Lemma \ref{theo:D1classicalcomm} we showed that classical S-CC$\mathcal{R}(\mathcal{G}^{(n,\omega)})$ is $\log_2 |\mathcal{V}|$ bit when the graph $\mathcal{G}^{(n,\omega)}$ satisfies {\bfseries (G0)-(G1)}. It is the minimum communication required for simultaneously satisfying ({\bfseries T0})-({\bfseries T1}) in this case. Here we consider the class of graphs $\mathcal{G}^{(n,\omega)}$ which satisfies the constraint ({\bfseries G0})-({\bfseries G2}) and has faithful orthogonal representation in minimum dimension $\omega$.  We first show that if we bound classical communication to $\log_2 \omega$ bit and allow classical public coin then one can satisfy ({\bfseries T0})-({\bfseries T1}) and achieve optimal payoff $\mathcal{P}_{\mathcal{R}(\mathcal{G}^{(n,\omega}))}^{*}$ for relation reconstruction (See Obs. \ref{obs:classical with public coin}). We then calculate the minimum amount of classical public coin assistance required to satisfy ({\bfseries T0})-({\bfseries T1}) and achieve the optimal payoff for the reconstruction of relation $\mathcal{R}(\mathcal{G}^{(n,\omega)})$ ({\bfseries T2}).
\begin{obs}
\label{obs:classical with public coin}
Given a graph $\mathcal{G}^{(n,\omega)}$, the strategy with only $\log_2 \omega$ bit classical communication for satisfying ({\bfseries T0}) is based on Alice and Bob finding a suitable deterministic strategy,  i.e an $n\omega \times n\omega$ table of conditional probabilities $p(b| C_x,C_y,a)$ given as $M$ at the beginning, which can be expressed as a $\omega \times n\omega$ table after compression. In the public coin-assisted scenario, Alice and Bob prepare all such deterministic strategies (or tables) each of which satisfies consistent labelling of cliques ({\bfseries T0}) before the game begins and index these tables. Using public coins, they alternated between these tables in different runs. Over multiple runs, they can satisfy ({\bfseries T1}). Trivially, they could use a classical public coin of the order of the total number of such deterministic strategies where each satisfies consistent labelling of the cliques ({\bfseries T0}).
\end{obs}

For example, consider the graph shown in Fig. \ref{fig:consistentcolouring} or the left graph of Fig. \ref{fig:SRdto2}, we have shown one classical deterministic strategy as represented through Table \ref{table:d3n2_T0} in Appendix \ref{app:prooftheo1}. Similarly, Alice and Bob could use another strategy represented by another Table \ref{table:d3n2_T0_2}. If Alice and Bob use $1$ bit of unbiased classical public coin to choose between Table \ref{table:d3n2_T0} and Table \ref{table:d3n2_T0_2}, they effectively are using the strategy given in Table \ref{table:d3n2_T0_3} which satisfies ({\bfseries T0}) as well as ({\bfseries T1}) and obtain the optimal payoff for this graph $\mathcal{P}_{\mathcal{R}(\mathcal{G}^{(2,3)})}^*=0.5$ since they fill all the entries $*$ with $0.5$. Now, we provide a lower bound on the amount of classical public coin required by Alice and Bob, when they are allowed to communicate $\log_2 \omega$ bit, to accomplish relation reconstruction.

\begin{table}[h]
\begin{center}
\begin{tabular}{|c c|ccc|ccc| } 
\hline
& & &$C_1$& & &$C_2$& \\
 & & $b=0$&$b=1$&$b=2$ & $b=0$&$b=1$&$b=2$\\
 \hline
     & $a=0$ & $1$ &$0$ & $0$ & $0$& $1$ & $0$ \\
$C_1$&$a=1$  &$0$ & $1$ &$0$ & $0$& $0$ & $1$ \\
     &$a=2$  &$0$ & $0$ &$1$ & $1$& $0$ & $0$ \\
 \hline
     &$a=0$ & $0$& $0$ & $1$ & $1$& $0$ &$0$  \\
$C_2$&$a=1$ & $1$& $0$ & $0$ &$0$ & $1$ &$0$  \\
     &$a=2$ & $0$& $1$ & $0$ &$0$ & $0$ &$1$\\
 \hline
\end{tabular}
    \caption{Another classical deterministic strategy for graph in Fig. \ref{fig:consistentcolouring}.}
\label{table:d3n2_T0_2}
\end{center}
\end{table}
\begin{table}[h]
\begin{center}
\begin{tabular}{|c c|ccc|ccc| } 
\hline
& & &$C_1$& & &$C_2$& \\
 & & $b=0$&$b=1$&$b=2$ & $b=0$&$b=1$&$b=2$\\
 \hline
     & $a=0$ & $1$ &$0$ & $0$ & $0$& $0.5$ & $0.5$ \\
$C_1$&$a=1$  &$0$ & $1$ &$0$ & $0$& $0.5$ & $0.5$ \\
     &$a=2$  &$0$ & $0$ &$1$ & $1$& $0$ & $0$ \\
 \hline
     &$a=0$ & $0$& $0$ & $1$ & $1$& $0$ &$0$  \\
$C_2$&$a=1$ & $0.5$& $0.5$ & $0$ &$0$ & $1$ &$0$  \\
     &$a=2$ & $0.5$& $0.5$ & $0$ &$0$ & $0$ &$1$\\
 \hline
\end{tabular}
\caption{Effective classical strategy with classical public coin for the graph in Fig. \ref{fig:consistentcolouring}.}
\label{table:d3n2_T0_3}
\end{center}
\end{table}

\begin{theo}\label{theo:dntod2nlowerbound}
  Given a graph $\mathcal{G}^{(n,\omega)}$ satisfying conditions {(\bfseries G0)}-{(\bfseries G2)} with faithful orthogonal range $\omega$, the lower bound on the amount of classical public coin assistance to $\log_2 \omega$ bit communication required for reconstruction of relation $\mathcal{R}(\mathcal{G}^{(n,\omega)})$ (and obtain optimal payoff) is equal to the minimum amount of classical public coin required for the same task when one bit communication is allowed and we consider another graph $\mathcal{G}^{(n,\omega=2)}$ with $n$ disconnected maximum cliques.
\end{theo} 

\begin{proof}
    See appendix \ref{app:prooftheo5} for the proof.
\end{proof}

We now provide the explicit lower bounds on classical public coin required for the reconstruction of relation $\mathcal{R}(\mathcal{G}^{(n,\omega)})$ as a function of the number of maximum cliques in the graph.

\begin{cor}
\label{corollary:necessary SR}
 Given a graph $\mathcal{G}^{(n,\omega)}$ satisfying {(\bfseries G0)}-{(\bfseries G2)} with faithful orthogonal range $\omega$, it is necessary (but may not be sufficient) to share classical public coin with $n$-inputs (i.e. $\frac{1}{n}\sum_{i=1}^{n}\left(\ket{ii}\bra{ii}\right)$) while communicating an $\omega$-level classical system for the reconstruction of relation $\mathcal{R}(\mathcal{G}^{(n,\omega)})$.
\end{cor}

\begin{proof}
    See appendix \ref{app:proofcorollaary1} for the proof.
\end{proof}

 Now, we will show that the lower bound on the amount of classical public coin required for achieving optimal payoff for relation ($\mathcal{R}(\mathcal{G}^{(n,\omega)})$) reconstruction while communicating $\omega$-level classical system is related to the existence of some specific kinds of Orthogonal Arrays. Before moving forward, we first introduce Orthogonal Arrays.
\begin{defi}\label{def:OA}
An $N\times k$ array $A$ with entries from set $S$ is called an orthogonal array $OA(N,k,s,t)$ with $s$ levels, strength $t (\in\{0,1,\cdots, k\})$ and index $\lambda$ if every $n\times t$ sub-array of $A$ contains each $t$-tuples based on $S$ appearing exactly $\lambda$ times as a row \cite{Hedayat1999}.
\end{defi}
Orthogonal Arrays have found interesting connections with absolutely maximally entangled states \cite{AME2015}, multipartite entanglement \cite{Goyeneche2016,Goyeneche2014}, quantum error-correcting codes \cite{Pang2022} etc. Here, we will consider orthogonal arrays $OA(N,k,s,t)$ where $t=2$ and $s=2$ and $S=\{0,1\}$. Let $T_k$ be the minimum $N$ for a fixed $k$ such that $OA(N=T_k,k,s=2,t=2)$ is an orthogonal array with $S=\{0,1\}$. Thus, in $OA(N=T_k,k,s=2,t=2)$ every $T_k\times 2$ sub-array has the tuples $\{(0,0),(0,1),(1,0),(1,1)\}$ appearing equal number of times as rows. \\

$T_n$ is related to the amount of classical public coin necessary and sufficient for the reconstruction of the relation $\mathcal{R}(\mathcal{G}^{(n,\omega=2)})$ with optimal payoff $\mathcal{P}_{\mathcal{R}(\mathcal{G}^{(n,2}))}^{*}$ when $\log_2 \omega$ bit classical communication is allowed from Alice to Bob.

\begin{cor}\label{cor:SR2}
Given a graph $\mathcal{G}^{(n,\omega)}$ satisfying {(\bfseries G0)}-{(\bfseries G2)} with faithful orthogonal range $\omega$, it is necessary (but may not be sufficient) to share classical public coin with $2$-inputs (for $n=2$) and $\log_2 T_{n-1}$-inputs (for $n>2$) while communicating an $\omega$- level classical system for relation $\mathcal{R}(\mathcal{G}^{(n,\omega)})$ reconstruction with optimal payoff $\mathcal{P}_{\mathcal{R}(\mathcal{G}^{(n,\omega}))}^{*}$.
\end{cor}

\begin{proof}
    See appendix \ref{app:proofcorollaary2} for the proof.
\end{proof}

Now we show that there exist some graphs $\mathcal{G}^{(n,\omega)}$ for which Alice and Bob need classical public coin while communicating $\omega$ level quantum or classical system for relation $\mathcal{R}(\mathcal{G}^{(n,\omega)})$ reconstruction with optimal payoff. As a consequence of this result, there are graphs for which $1$ bit classical communication when assisted by a finite amount of classical public coin can be powerful compared to $1$ qubit quantum direct communication resources when considering this particular task and payoff.

\begin{theo}\label{theo:qudit}
There exist graphs $\mathcal{G}^{(n,\omega)}$ satisfying {(\bfseries G0)}-{(\bfseries G2)} and faithful orthogonal range $\omega$, such that while using $\omega$ dimensional classical or quantum channel, the assistance of public coins is necessary for relation $\mathcal{R}(\mathcal{G}^{(n,\omega)})$ reconstruction and obtaining optimal payoff $\mathcal{P}_{\mathcal{R}(\mathcal{G}^{(n,\omega)})}^{*}$.
\end{theo}

\begin{figure}[h]
    \centering
    \begin{DiagramV}{0}{0}
\begin{move}{0,0} 
\fill[black] (0,0) circle (0.25);
\fill[black] (6,0) circle (0.25);
\draw (0,0) -- (6,0);
\draw (0,0-1) node {$v_1$};
\draw (6,0-1) node {$v_2$};
\draw (3,+1.5) node { $C_1$};
%
\fill[black] (0+8,0) circle (0.25);
\fill[black] (6+8,0) circle (0.25);
\draw (0+8,0) -- (6+8,0);
\draw (0+8,0-1) node {$v_3$};
\draw (6+8,0-1) node {$v_4$};
\draw (3+8,+1.5) node {$C_2$};
%
\fill[black] (15.5,0) circle (0.1);
\fill[black] (17,0) circle (0.1);
\fill[black] (18.5,0) circle (0.1);
\fill[black] (0+20,0) circle (0.25);
\fill[black] (6+20,0) circle (0.25);
\draw (0+20,0) -- (6+20,0);
\draw (0+20,0-1) node {$v_{2n-1}$};
\draw (6+20,0-1) node {$v_{2n}$};
\draw (3+20,+1.5) node {$C_n$};
\end{move}
\end{DiagramV}
    \caption{Example for a graph $\mathcal{G}^{(n,\omega)}$ satisfying Theorem \ref{theo:qudit} with $n$ disconnected maximum cliques of size $\omega=2$. }
    \label{fig:cbitSRvsqubit}
\end{figure}

\begin{proof}
Assume that Alice is allowed to communicate an $\omega$-dimensional system to Bob. We prove the above-mentioned theorem by showing the existence of a graph that satisfies the claim. Let us consider the graph $\mathcal{G}^{(n=\omega+2,\omega)}$ satisfying {(\bfseries G0)}-{(\bfseries G2)} and having faithful orthogonal representation in minimum dimension $\omega$ where any two the maximum size cliques are disconnected. For an example, see Fig. \ref{fig:cbitSRvsqubit} where $\omega=2$. 

Note that for such a graph, the maximum payoff achievable by communicating $\log \omega$ qubit, $\mathcal{P}_{\mathcal{R}(\mathcal{G}^{(\omega+2,\omega)})}$, is always less than the optimal payoff $\mathcal{P}_{\mathcal{R}(\mathcal{G}^{(n,\omega}))}^{*}=\frac{1}{\omega}$. This is because, only $\omega+1$ mutually unbiased bases (MUBs) are possible in $\C^{\omega}$, which can be used to encode and decode in an unbiased way, a maximum of $\omega+1$ maximum cliques in the considered graph.
If Alice is allowed to send $\log_2 \omega$ bit without having access to public coins, then the payoff obtained is zero (see Lemma \ref{theo:D1classicalcomm}).
On the other hand, by using finite classical public coins, all the deterministic strategies using $\log_2 \omega$ bit which satisfy ({\bfseries T0}) (which are finite in number) can be mixed to obtain the optimal payoff $\mathcal{P}_{\mathcal{R}(\mathcal{G}^{(n,\omega}))}^{*}=\frac{1}{\omega}$. 

\end{proof}
For the graph in Fig. \ref{fig:cbitSRvsqubit}, the necessary and sufficient amount of classical public coins to achieve $\mathcal{P}_{\mathcal{R}(\mathcal{G}^{(n,2)})}=\frac{1}{2}$ while communicating $1$ bit is given in Corollary \ref{cor:SR2}. Also, the maximum payoff achieved when $1$ qubit is communicated from Alice to Bob is upper bounded by $\frac{1}{2}$ (the optimal payoff $\mathcal{P}_{\mathcal{R}(\mathcal{G}^{(n,2}))}^{*}$ can be achieved only for $n\leq 3$). Thus, $1$ bit classical communication when assisted by a finite amount of classical public coin can outperform $1$ qubit quantum direct communication resources when considering this task.

\subsubsection{Classical communication assisted by quantum public coin}\label{subsubsec:corrq}
At this point, a natural question is whether quantum correlations (quantum public coin) can enhance classical communication more than classical public coin. In the following theorem, we mention an instance where this is the case.

\begin{theo}\label{theo:cshared}
 For public coin assisted classical communication, there exist graphs $\mathcal{G}^{(n,\omega)}$ satisfying conditions ({\bfseries G0})-({\bfseries G2}), such that the separation between classical and quantum public coins required for reconstruction of relation $\mathcal{R}(\mathcal{G}^{(n,\omega)})$ is unbounded.
\end{theo}

\begin{proof}
Let us consider the graph $\mathcal{G}_{disc.}^{(n,\omega)}$ given by $n$ disjoint maximum cliques of size $\omega= 2$. $1$ bit classical communication assisted by $n-1$ input classical public coin gives payoff $0$ (see Corollary \ref{corollary:necessary SR}). On the other hand, when assisted by $1$ e-bit of entanglement (a two-qubit maximally entangled state), Alice chooses $n$ distinct orthogonal pairs of states from the equatorial circle of the Bloch sphere corresponding to the $n$ possible input maximum cliques. Now Alice and Bob perform the protocol the same as remote state preparation \cite{pati2000minimum,bennett2001remote}, which allows perfect transmission of the states from an equatorial circle of the Bloch sphere with  1 e-bit of shared entanglement and 1 bit of classical communication. After successful transmission of the state, Bob performs qubit projective measurement based on his input $C_y$ along one of the bases chosen by Alice. This makes the payoff $\mathcal{P}_{\mathcal{R}(\mathcal{G}_{disc.}^{(n,2)})}>0$. Thus increasing $n$ will require an increasing amount of classical public coin, while  1 e-bit of entanglement (quantum public coin) ensures a quantum protocol to achieve a non-zero payoff. 
\end{proof}
For example, the symmetric choice of $n=4$ directions on the Bloch sphere implies that this protocol can achieve
 $\mathcal{P}_{\mathcal{R}(\mathcal{G}_{disc.}^{(4,2)})}=\sin^2(\frac{\pi}{8})\approx 0.1464.$  

\subsection{Summary of results}
In this Section \ref{sec:sccind_results} we have presented several results. Here, we highlight the main results, summarised in the form of the following two tables \ref{table:results 1} and \ref{table:result 2}.  First, in Table \ref{table:results 1} we have summarised the classical and quantum CCR and S-CCR without public coin assistance for different graphs when considering the relation $\mathcal{R}(\mathcal{G}^{(n,\omega)})$. This table summarises the main results of this work where we quantify an unbounded quantum advantage in S-CC$\mathcal{R}(\mathcal{G}^{(n,\omega)})$ for some class of graphs. Next, in Table \ref{table:result 2}, we summarise our result on the amount of public coin assistance required for restricted quantum and classical direct communication resources. Here, we have also mentioned our result on the unbounded advantage of sharing quantum public coins over sharing classical public coins when using only $1$ bit direct communication. 

\begin{widetext} 
\begin{center}
\begin{table} [H]
\begin{center}
\begin{tabular}{c c cc c}
\hline
 \multirow{2}{*}{ {\bf Communication Task} } & \multicolumn{2}{c}{{\bf Resource Comparison} } & {\bf Quantum} & \multirow{2}{*}{ {\bf Ref.} }   \\
  &    {\bf Classical}      & {\bf Quantum}  &  {\bf Advantage}             &           \\\hline
   Distributed computation of $\mathcal{R}(\mathcal{G}^{(n,\omega)})$                                & $\log_2\omega$ bit   & $\log_2\omega$ qubits     & $\Omega(1)$    & Section \ref{CLP:classical}, Theorem \ref{theo:D0classicalcomm}                 \\
  Reconstruction of $\mathcal{R}(\mathcal{G}^{(n,\omega)})$                             & $\log_2|\mathcal{V}|$ bit  & $\log_2d_{\C}$ qubits &  $\Omega(\log_2 |\mathcal{V}|)$          & Section \ref{subsubsec:directq}, Theorem \ref{theo:qadvantage}               \\\hline
        Reconstruction of $\mathcal{R}(\mathcal{G}_{disc.}^{(n,\omega)})$                             & $\log_2n\omega$ bit  & $\log_2\omega$ qubits    &   $\Omega(\log_2 n)$    & Section \ref{subsubsec:directq}, Theorem \ref{theo:unbounded}                \\
 Reconstruction of $\mathcal{R}(\mathcal{G}_{NNCC(r)}^{(n,\omega)})$                             & $\log_2(n(\omega-r)+r)$ bit  & $\log_2\omega$ qubits     &  $\Omega(\log_2 n)$    & Section \ref{subsubsec:directq}, Theorem \ref{theo:unbounded}                \\
       Reconstruction of $\mathcal{R}(\mathcal{G}_{Paley(q)})$                             & $\log_2q$ bit  & $\log_2\frac{q+1}{2}$ qubits     &   $\Omega(1)$   & Theorem \ref{theo:qadvantage} \& Section \ref{SCLP:other}                  \\\hline
\end{tabular}
\caption{Resource comparison for classical vs quantum one-way communication tasks, {\it i.e.} distributed computation and relation reconstruction while considering $\mathcal{R}(\mathcal{G}^{(n,\omega)})$, with some examples of quantum advantage in S-CC$\mathcal{R}(\mathcal{G}^{(n,\omega)})$ for certain families of graphs considered in Section \ref{subsubsec:directq}\label{table:results 1}.}
 \end{center}
\end{table}
\end{center}
\end{widetext}

\begin{widetext} 
\begin{center}
\begin{table} [H]
\begin{center}
\begin{tabular}{c c cc c}
\hline
\multirow{2}{*}{ {\bf Resource Constraint} }             & \multirow{2}{*}{ {\bf Communication Task} } & \multicolumn{2}{c}{{\bf Resource Comparison} } & \multirow{2}{*}{ {\bf Ref.} }   \\
  &   & ~{\bf Only Classical}~      & ~{\bf Quantum allowed}~                 &           \\\hline
  \hline 
One-way Communication  &  Reconstruction of $\mathcal{R}(\mathcal{G}^{(n,\omega)})$                                & $\log_2\omega$ bit & $\log_2\omega$ qubits          & Section \ref{subsec:corr},               \\
+ Classical Public Coin (SR)  &     & + $\log_2n$ bit SR*   & + No SR required          &    Corollary \ref{cor:SR2}         \\\hline
One-way Communication                 & Reconstruction of $\mathcal{R}(\mathcal{G}^{(n,\omega=2)})$                                & 1 bit   & 1 bit               & Section \ref{subsec:corr},                \\
+ Public coins &    &+ $\log_2 n$ bit SR   &+ $1$ EPR pair               &      Theorem \ref{theo:cshared}       \\
\hline
\end{tabular}
\caption{Resource Comparison for the communication task considered in Section \ref{subsec:corr} where we allow public coins and compare purely classical protocols with hybrid protocols allowing some quantum resource --- communication (first row) or entanglement (second row).\\
*Here the $\log_2n$ bit classical public coin allow relation $\mathcal{R}(\mathcal{G}^{(n,\omega)})$
reconstruction but does not always achieve the optimal payoff $\mathcal{P}_{\mathcal{R}(\mathcal{G}^{(n,\omega}))}^{*}$. The classical public coin necessary for achieving $\mathcal{P}_{\mathcal{R}(\mathcal{G}^{(n,\omega)})}^{*}$ is connected to the problem of  orthogonal arrays.}
\label{table:result 2}
\end{center}
\end{table}
\end{center}
\end{widetext}

\section{Applications}\label{sec:appl}
In this section, we discuss some useful applications of the relation reconstruction task. The first application, in Section \ref{appli:mub}, is the operational detection of MUBs from the observation of the statistics. We consider some specific type of graph $\mathcal{G}$ with both maximum clique size and faithful orthogonal representation in minimum dimension $\omega$. If a quantum strategy using a $\omega$ level quantum system can achieve the upper bound of the optimal payoff (that is $\mathcal{P}_{\mathcal{R}(\mathcal{G})}^Q=\mathcal{P}_{\mathcal{R}(\mathcal{G})}^*$) for such a graph $\mathcal{G}$, then Bob must have used measurements corresponding to MUBs for decoding. In the next application, in Section \ref{subsec:sdi}, we consider the problem of detecting the non-classical resources in both direct communication and in the shared correlation (black-box) scenario. Additionally, we discuss an application of the communication task as a dimension witness. In the following, we discuss each of the applications in greater detail.

\subsection{Detecting Mutually Unbiased Bases}\label{appli:mub}
We show the operational detection of MUBs from the observation of the statistics of our communication task, showing that quantumly achieving the payoff $\mathcal{P}_{\mathcal{R}(\mathcal{G}^{(n,\omega)})}^*$ ({\bfseries T2}) for some graphs implies the detection of MUBs.

A pair of projective measurements for a $d$-dimensional Hilbert space are mutually unbiased if the squared
length of the projection of any basis element from the first onto any basis element of the second is
exactly $1/d$. Mutually unbiased bases (MUBs) are found to be optimal in several information-theoretic tasks and also in quantum cryptography \cite{Ivonovic1981,Wootters1989,PhysRevA.75.022319,PhysRevLett.92.067902,Aravind2003,Englert2001}. 
\begin{obs}
Consider a graph consisting of $n$ maximum cliques of size $\omega$ that are completely disconnected from each other --- $\mathcal{G}_{disc.}^{(n,\omega)}$. This graph has faithful orthogonal representation in dimension $d_{\R}=d_{\C}=\omega$. If a quantum strategy with direct communication of an $\omega$-level system can achieve the optimal payoff i.e. $\mathcal{P}_{\mathcal{R}(\mathcal{G}_{disc.}^{(n,\omega)})}=(\frac{1}{\omega}=\mathcal{P}_{\mathcal{R}(\mathcal{G}_{disc.}^{(n,\omega)})}^*)$  for the relation reconstruction task, then the measurements performed by Bob must be those corresponding to MUBs.
\end{obs}
For example, let us consider one such graph, which allows for the detection of qubit-MUBs. The simplest graph consists of three maximum cliques of size $\omega=2$ that are disconnected from each other, $\mathcal{G}_{disc.}^{(n=3,\omega=2)}$. Upon receiving her input maximum clique and clique label, Alice prepares her state in one of the pairs of the eigenstates of three qubit-MUBs corresponding to the disjoint maximum cliques of this graph and sends the qubit to Bob. Bob performs his measurement corresponding to one of the above three MUBs based on his input maximum clique. Evidently in this case, the payoff turns out to be $\mathcal{P}_{\mathcal{R}(\mathcal{G}^{(3,2)})}=\frac{1}{2}$. Conversely, one can see that to achieve the optimal payoff it is required to produce the prepare and measure probabilities corresponding to the disconnected pairs of vertices of the graph completely unbiased.

\subsection{Semi-Device Independent Detection of Non-Classical Resources and Dimension Witness}\label{subsec:sdi}
In a {\it prepare and measure} setup, which underlies several information-theoretic tasks, two prime questions of practical interest are- (i) is the transmitted system (alternatively, are the prepare and measure devices) {\it non-classical}? and (ii) what is the operational dimension of the transmitted system? For quantum systems, the second question reduces to finding a lower bound on the Hilbert space dimension, {\it i.e.} to find a {\it dimension witness} \cite{Brunner2008,Brunner2013,Ahrens2014,Cai2016}. If these questions are answered based on the {\it input-output} probability distribution $\{P(b|x,y)\}$, where $x\in X$ and $y\in Y$ are inputs and $b\in B$ is the output, without referring to any information about the encoding and decoding devices, the protocol is {\it device independent}. If partial information about the devices is available, the scenario is called {\it semi-device independent}. In the following, we show that the proposed relation reconstruction task can be used as a semi-device independent witness of non-classicality and dimension.

While answering the first question, we will consider two scenarios, first, where no public coin is available. This scenario allows us to determine the non-classicality of the transmitted system. Second, where only a finite amount of public coins are available and a classical bit has been transmitted, allows us to answer whether the public coin is non-classical or not. For both cases, let us consider the two distant parties executing the relation reconstruction task with a class of graphs satisfying conditions condition ({\bfseries G0})-({\bfseries G1}). Now, in the first case let us also assume that it is known that the operational dimension of the transmitted system is strictly upper bounded by $|\mathcal{V}|$, the number of vertices of the graph. If the distant parties can achieve a non-zero payoff (calculated from $P(b|x,y)$ according to the definition in Eq.\ref{eq:pay}), it follows from Theorem \ref{theo:D1classicalcomm} that the transmitted system is non-classical. In the second case with a finite public coin and a classical bit communication, let us consider the graph $\mathcal{G}^{(n,\omega=2)}$ with all disconnected maximum cliques. This graph has a faithful orthogonal range $\omega=2$. If the local dimension of the public coin is strictly upper-bounded by $n$, the number of maximum cliques in the graph. Then (see the example in the proof of Theorem \ref{theo:cshared}) payoff $\mathcal{P}_{\mathcal{R}(\mathcal{G}^{(n,2)})}>0$ implies that the public coin is non-classical.          
To answer the question about dimension witness we first observe the following: it follows from Theorem \ref{theo:D0classicalcomm} that given a graph with $n$ number of maximum cliques of size $\omega$, performing the distributed computation of $\mathcal{R}(\mathcal{G}^{(n,\omega)})$ requires at least $\omega$-level system needs to be communicated from Alice to Bob. This fact applies to any arbitrary graph. 
Even in the presence of Public coins, if Alice's encoding and Bob's decoding can perform this task without any error, it will imply that the communicated system must have an operational dimension of at least $\omega$. 

\section{Summary \& Discussions}\label{sec:disc}
In a non-asymptotic {\it prepare and measure} scenario, the problem of efficient encoding of classical information in a quantum system has been a topic of interest in recent times \cite{Frenkel2015, Heinosaari2020, Pauwels2022, tavakoli2022, Patra2022, halder2022}. Communication complexity, a prototype of distributed computing, measures the efficiency of such an encoding by the separation between the operational dimension of the classical and quantum message. A large separation for some distributed computation tasks demonstrates the advantage of quantum communication resources over classical ones. The present work proposes one such task, called reconstruction of relations specified by the distributed clique labelling problem. The key significance of this work is that (i) it presents a novel case of unbounded quantum advantage in the communication complexity type task and (ii) the size of the input for the task is the same as the order of the quantum advantage. All previous results showing unbounded advantages in communication scenarios with well-defined input data at local stations \cite{Massar01,Perry15,Liu16} required input data to be exponentially larger than the advantage and as such are hard to practically demonstrate. 

In this relation reconstruction task, we show a class of graphs for which the separation between the dimension of quantum and classical systems necessary can be unbounded without public coins or pre-shared between the parties. In the presence of public coins, however, this separation disappears. While quantum communication does not require public coins, the amount of classical public coin assistance that is necessary (but may not be sufficient) for classical communication to accomplish the task scales linearly with the number of maximum cliques. Additionally, we also show that a $1$ e-bit assisted classical $1$ bit channel performs a task that would otherwise require the assistance of a $1$ bit channel and an unbounded amount of classical public coin.  

The present work can be seen as an addition to the earlier attempts to demonstrate the separation of classical and quantum communication complexity with relations \cite{Buhrman98,Raz,Bar2008,Gavinsky2007}. For example, Buhrman, Cleve, and Wigderson in \cite{Buhrman98} showed an exponential gap between classical and quantum zero-error communication complexity for a promise problem in the presence of public coins. Later Raz in \cite{Raz} showed that an exponential gap in communication exists for communication complexity of a relation while considering bounded-error public coin assisted interactive protocols. Bar-Yossef {\it et al.} in \cite{Bar2008} showed an exponential separation for one-way protocols as well as simultaneous protocols with public coins for a relational problem, called the Hidden Matching Problem. In a tripartite setup (two players and a referee), well known as {\it quantum fingerprinting}, Buhrman {\it et al.} reported an exponential advantage of quantum communication \cite{Buhrman01}. Based on the result of \cite{Buhrman98}, in a slightly different setup (where the sender also produces an output), an exponential gap has also been demonstrated in the task of simulating statistics of maximally entangled states \cite{Brassard99}. Moreover, when quantum inputs are considered, the generalized statistics simulation task becomes classically impossible \cite{Rosset13}.

The question of an unbounded separation has also been addressed in the literature. Perry {\it et al.} \cite{Perry15} showed an unbounded classical vs quantum separation in terms of {\it internal information cost}, defined as the amount of Alice’s input information revealed to Bob, while Bob’s task is to exclude certain combinations of bits that Alice might have. 
Although later Liu {\it et al.} \cite{Liu16} proved that the quantum communication complexity of those tasks scales at least logarithmically in the input string length, indicating that quantum advantage in this setting is specific to the internal information cost/complexity rather than communication complexity. 
Beyond the conventional setting of communication complexity, Galvao {\it et al.} \cite{Galvao03} considered a problem where a system interacts with a classical field while travelling between two points and a decoding measurement at the destination has to answer a {\it yes-no} question about the intermediate field. This problem considers a scenario which differs from communication complexity because the concept of classical input is of a different type since it is the continuous field that is capable of coupling to the qubit (eg, magnetic one). It does not have a character of binary data provided to Alice and Bob separately, which needs to be communicated. It is rather the binary classical characteristic of a continuously parameterised quantum/classical channel between them that is to be detected. They proved an unbounded gap between the sizes corresponding to classical and quantum systems required to answer correctly. But the most relevant unbounded separation example to the present work was provided by Massar {\it et al.} in \cite{Massar01}. They showed that the quantum communication complexity of \texttt{NOT-EQUAL} problem is a qubit and the classical communication complexity increases with the input size. Considering the similar spirit of this result and our results, a comparison is in order. While in our task the required classical communication grows as the logarithm of the cardinality of the input set, classical communication complexity for the task in \cite{Massar01} scales as the double logarithm of the cardinality of the input set. This implies a challenge for practical demonstration of the unbounded advantage in \cite{Massar01} as the size of the input set grows exponentially with the advantage. Whereas in our problem, this scaling is only linear. Second, our task shows an unbounded advantage for any $d$ dimensional quantum system, while the protocol in \cite{Massar01} works only for qubits 
and generalization for higher dimensions does not seem straightforward.

Another important aspect of the present work is that the relations considered here are given by orthogonality graphs. A similar approach while demonstrating the advantage of quantum communication over classical was taken by Saha {\it et al.} in \cite{Saha2019}. The authors in \cite{Saha2019} considered a graph colouring task, called vertex equality problem, executed by two spatially separated parties. They showed that quantum advantage in one-way communication appeared whenever a class of graphs, called state-independent contextuality graphs (SIC graphs) were considered. In contrast, the quantum advantage in the communication task proposed in this article can be observed independent of the usefulness of the graphs in demonstrating state-independent contextuality. Therefore, in our case, the quantum advantage in one-way communication can not be attributed to contextuality. Interestingly, in \cite{Saha2023} the authors showed that the quantum separation for computation of a partial function via communication task based on state-independent contextuality witnesses can be polynomially large, whereas our task, independent of contextuality, can obtain an unbounded separation for the computation of a relation via communication.

In a practical setting, one may not always have the same input sets for both parties. One straightforward direction for a generalisation of relation reconstruction would entail, one party (Alice) receiving input maximum cliques and their label over some graph $\mathcal{G_A} \subset \mathcal{\Tilde{G}}$ while the other party (Bob) receives input maximum clique that is to be labelled from some other graph $\mathcal{G_B}\subset \mathcal{\Tilde{G}}$ only to be consistent with the label of Alice if $\mathcal{G_A} \cap \mathcal{G_B}=\mathcal{G} \neq \emptyset$. As long as there is a quantum advantage in the relation reconstruction task for distributed clique labelling problem over this graph $\mathcal{G}$, one could still find the usefulness of the relation reconstruction task in the setting with different input sets.

This work leaves several questions open. For example, could there be a task such that the scaling of classical vs quantum communication with binary colourable graphs be exponential in the presence/absence of public coins (possibly for two-way communication complexity)? Could one obtain a linear scaling when the two parties compute a function instead of a relation? Besides these general questions, there are some particular points about the present study that remain unresolved. First, does the unbounded separation between classical and quantum communication persist when one departs from the zero-error scenario considered in this work and considers a degree of error in the computation? In Section \ref{subsubsec:corrc}, the connection between a lower bound to the amount of classical public coin in the bounded communication setting and orthogonal arrays, shows that given arbitrary graphs with a large number of maximum cliques, finding such a lower bound is a hard problem. In Section \ref{subsubsec:corrq} the advantage of using a quantum public coin(entanglement) instead of a classical public coin (shared randomness) to assist a bounded classical communication has been demonstrated by achieving a higher payoff. However, it remains unknown what is the optimal payoff for the entanglement-assisted case. A monotonically decreasing payoff with the increasing number of maximum cliques ($n$) might suggest a limit of this advantage. Finally, in the applications section (Section \ref{appli:mub}) the robustness of the scheme to detect MUBs is not known.

Finally, one can look at the present protocol from a foundational perspective. Namely, it can be seen as a qualitative simulation of the quantum statistics on demand. The relation-reconstruction task proposed in this article could bridge the gap between conventional communication complexity and sampling problems with communication \cite{Ambainis2003, Watson2020}. Precisely, in our protocol, the spatially separated parties are given some set of favourable events and it is required that the events be quantitatively simulated by classical communication so that all of them occur with nonzero probability like it is in the quantum case. Considering the class of graphs in Fig.\ref{fig:cbitSRvsqubit}, obtaining the non-zero value of the payoff function (Eq.\ref{eq:payoff}), reduces to the simulation of the set of correlations generated from pure qubit states and qubit projective measurements. In order to simulate the prepare-measure statistics of a qubit, the authors in \cite{tavakoli2022} show that it is necessary and sufficient to communicate two classical bits when the parties are assisted by pre-shared randomness. In the same spirit, Corollary \ref{corollary:necessary SR} says that it is necessary to share an unbounded amount of randomness between the sender and receiver besides one bit of classical communication to simulate the statistics of qubits for the class of graphs as in Fig.\ref{fig:cbitSRvsqubit}.  

Looking at the protocol from yet another angle, we can see it as a distribution of a (conditional) randomness with the help of a restricted communication channel. 
This raises the question of the possible relation of the present scheme to discrete analogues of bosonic sampling \cite{Bryan2015}. Quantum advantage in the latter case relies on  
the hypotheses of the computational hardness of some classical tasks. It would be interesting to see whether additional graph structure and modification of the present protocol could imply exponential separation in sampling that would not rely on hypotheses of this type.

\section{Acknowledgements}
S.R. acknowledges Markus Grassl for discussion on Orthogonal Arrays. The ’International Centre for Theory of Quantum Technologies’ project (contract no. MAB/2018/5) is carried out within the International Research Agendas Programme of the Foundation for Polish Science co-financed by the European Union from the funds of the Smart Growth Operational Programme, axis IV: Increasing the research potential (Measure 4.3).  P. H. also acknowledge that this work is partially carried out under IRA Programme, project no. FENG.02.01-IP.05-0006/23, financed by the FENG program 2021-2027, Priority FENG.02, Measure FENG.02.01., with the support of the FNP. R.R. acknowledges support from the Early Career Scheme (ECS) grant "Device-Independent Random Number Generation and Quantum Key Distribution with Weak Random Seeds" (Grant No. 27210620), the General Research Fund (GRF) grant "Semi-device-independent cryptographic applications of a single trusted quantum system" (Grant No. 17211122) and the Research Impact Fund (RIF) "Trustworthy quantum gadgets for secure online communication" (Grant No. R7035-21). S.S.B acknowledges funding by the Spanish MICIN (project PID2022-141283NB-I00) with the support of FEDER funds, and by the Ministry for Digital Transformation and of Civil Service of the Spanish
Government through the QUANTUM ENIA project call- Quantum Spain project, and by the European Union through
the Recovery, Transformation and Resilience Plan - NextGeneration EU within the framework of the Digital Spain 2026 Agenda.

\appendix

\section{Success Probability for Reconstruction of Relations}\label{app:rate}
Given a relation $\mathcal{R}\subseteq X \times Y \times B$ for the bipartite prepare and measure a scenario where $X$ and $Y$ are the set of inputs for Alice and Bob and $B$ is the set of outputs for Bob, we are interested in success probability $P_k(\mathcal{R})$ of relation reconstruction after $k$ number of rounds, where $k$ is large. Additionally, Alice and Bob's protocol is agnostic to the number of rounds. Every tuple $(x,y,b)\in \mathcal{R}$ must occur at least once in these $k$ rounds for the correct reconstruction of the relation $\mathcal{R}$. The cardinality $|\mathcal{R}|=\Gamma$ is the total number of all such events, which implies $k \geq \Gamma$ for reconstruction to be possible. Here we assume that the inputs are sampled from a uniform distribution. If Alice encodes her input $x\in X$ in the message $\tau_x$ in each round and Bob outputs $b\in B$ depending on his input $y\in Y$ and Alice's message,     
\begin{align}
    P(x,y,b)&= \sum_{\tau_x}P(x,y,b,\tau_x)\\
    &=\sum_{\tau_x} P(b|y,\tau_x)P(\tau_x|x)P(x)P(y)\\    P(x,y,b)&= P(b|y,x)P(x)P(y)
\end{align}
if $P(\tau_x|x)=1 \forall x\in X$ (this is the situation in the scenario when a pre-shared public coin is not allowed). We can consider a strict ordering of the elements in $\mathcal{R}$. Given this ordered list, we can define $\alpha(k)=\{\alpha_1,\alpha_2,\cdots,\alpha_{\Gamma}
\}$ where $\alpha_i\in \mathbb{N}$ is the frequency of occurrence of the $i^{th}$ element $(x_i,y_i,b_i)$ of ordered list $\mathcal{R}$ given $k$ number of rounds have occurred and thus $\sum_{i=1}^{\Gamma}\alpha_i= k ~\forall ~\alpha$. The instances favourable for successful reconstruction of relation correspond to the set of $\alpha(k)$ where each of the elements of $\mathcal{R}$ occur with non-zero frequency.
 The probability of reconstruction of $\mathcal{R}$ given $k$ number of rounds is thus given by the total probability of occurrences of the $\alpha(k)$ with the aforementioned property. 
\begin{align}
    P_k(\mathcal{R})&= \sum_{\alpha(k)}P(\alpha(k)|k)\nonumber\\
    &=\sum_{\alpha}P(\{\alpha_1,\alpha_2,\cdots,\alpha_{\Gamma}\}|k)\nonumber\\
    &=\sum_{\alpha}(\prod_{i=1}^{\Gamma}P^{\alpha_i}(x_i,y_i,b_i))
\end{align}
Since, $\forall~\alpha~\forall i\in\{1,2\cdots,\Gamma\},\alpha_i>0$, therefore, 
\begin{align}
P_k(\mathcal{R})&=\left(\prod_{i=1}^{\Gamma}P(x_i,y_i,b_i)\right)\left(\sum_{\alpha}(\prod_{i=1}^{\Gamma}P^{\alpha_i-1}(x_i,y_i,b_i))\right)\nonumber\\
\end{align}
Notice that if any of the terms $P(x_i,y_i,b_i)=0$ then the probability of successful reconstruction after $k$ rounds $P_k(\mathcal{R})$ becomes zero as well. Therefore,
\begin{equation}
P_k(\mathcal{R})\neq 0 \implies P(b|x,y) \neq 0~~ \forall (x,y,b) \in \mathcal{R}
\end{equation}

\noindent 
\textbf{Remark:} $P(b|x,y,\tau_x)=1$ $\forall x\in X,y \in Y$ such that  $\exists!b\in B$ satisfying $(x,y,b)\in \mathcal{R}$. For rest of the $(x,y,b)\in\mathcal{R}$, $P(b|x,y)\in(0,1).$

Now, we define $B_{x,y}=\{b\in B:(x,y,b)\in\mathcal{R}\}$ which is the set of all acceptable outputs for Bob given the input is $x$ and $y$ for Alice and Bob respectively. Then, $\sum_{b\in B_{x,y}}P(b|x,y,\tau_x)=1~~ \forall~~ B_{x,y}$.\\

We aim to maximise the success probability $P_k(\mathcal{R})$ in the scenario when Alice and Bob are not aware of the total number of rounds, say $k_{max}$, a priory and thus they should decide the probabilities of the events in $\mathcal{R}$ independent of $k_{max}$. To achieve this we start by using the Lagrange multiplier.\\

Now, in order to maximise the success probability of reconstruction for $k$ number of rounds we define 
\begin{align}
&L=P_k(\mathcal{R})-\sum_{B_{x,y}}\lambda_{B_{x,y}} \left(1-\sum_{b\in B_{x,y}}P(b|x,y)\right)
\end{align}

For $j^{th}$ element $(x_j,y_j,b_j)$ in ordered list of $\mathcal{R}$, 
\begin{align}
&\frac{\partial L}{\partial P(x_j,y_j,b_j)}=0\nonumber\\
\implies& \sum_{\alpha(k)}\left(\alpha_j P(x_j,y_j,b_j)^{-1}\right)\left(\prod_{i=1}^{\Gamma}P^{\alpha_i}(x_i,y_i,b_i)\right)\nonumber\\
&-\lambda_{B_{x_j,y_j}}\left(P(b_j|x_j,y_j)\right)\\
\implies& \lambda_{B_{x_j,y_j}}=\frac{\sum_{\alpha(k)}\alpha_j \left(\prod_{i=1}^{\Gamma}P^{\alpha_i}(x_i,y_i,b_i)\right)}{P^2(x_j,y_j,b_j)}
\end{align}
For a given $k$, the optimal probabilities $P(x_i,y_i,b_i)=P(b_i|x_i,y_i)P(x_i,y_i)$ can be calculated that yields maximum value of $P_k(\mathcal{R})$. However, for any arbitrary $k$, the expression of $\lambda_{B_{x_i,y_i}}$ is a function of $k$ as $\alpha(k)$ and $\alpha_i$ are a function of $k$. Since Alice and Bob do not have prior information about k, they have to agree on values of probabilities $P(x_i,y_i,b_i)$ independent of $k$. Thus, the obvious solution is $P(b|x,y)=constant$ $\forall~ 
 b\in B_{x,y}$. Here we assume that the inputs are sampled from a uniform distribution. This shows the necessity of our payoff function. Maximising the Payoff guarantees that the $P_k(\mathcal{R})$ is maximised to some local maxima.

\section{A concrete example}\label{app:example}

 Here we provide an example of a particular simple graph to help solidify the ideas of the clique labelling problem, relation reconstruction task for $\mathcal{R}(\mathcal{G}^{(n,\omega)})$ and the conditional probability table introduced in section \ref{sec:ccr}. Consider the graph $\mathcal{G}^{(n=2,\omega=3)}$ (see Fig. \ref{fig:consistentcolouring} for vertex indexing), with $n=2$ maximum cliques of size $\omega=3$ that share a common vertex. 

\begin{figure}[h]
    \centering
\begin{DiagramV}{0}{0}
\begin{move}{0,0} 
\fill[black] (0,0) circle (0.25);
\draw (0,0-1) node {$v_1$};
\fill[black] (6,0) circle (0.25);
\draw (6,0-1) node {$v_3$};
\fill[black] (3,5.196) circle (0.25);
\draw (3,5.196+1) node {$v_2$};
\draw (0,0) -- (6,0);
\draw (0,0) -- (3,5.196);
\draw (3,5.196) -- (6,0);
\draw (3,5.196/2 -0.5) node{$C_1$};
\fill[black] (3+6,5.196) circle (0.25);
\draw (3+6,5.196+1) node {$v_4$};
\fill[black] (6+6,0) circle (0.25);
\draw (6+6,0-1) node {$v_5$};
\draw (0+6,0) -- (6+6,0);
\draw (0+6,0) -- (3+6,5.196);
\draw (3+6,5.196) -- (6+6,0);
\draw (3+6,5.196/2 -0.5) node {$C_2$};
\end{move}
\end{DiagramV}    
    \caption{In this example, the graph $\mathcal{G}^{(2,3)}$ consists of two maximum cliques $C_1$ and $C_2$ of size $\omega=3$.}
    \label{fig:consistentcolouring}
\end{figure}

The mapping of binary colourings to clique labellings for maximum clique $C_1$ can be given by:
\begin{align}
    &f(v_1)=1, f(v_2)=f(v_3)=0 
    \implies  &g_{C_1}=0\nonumber\\
    &f(v_2)=1, f(v_1)=f(v_3)=0 
    \implies  &g_{C_1}=1\nonumber\\
    &f(v_3)=1, f(v_1)=f(v_2)=0 
    \implies  &g_{C_1}=2
\end{align}
Similarly, the mapping of binary colourings to clique labellings for maximum clique $C_2$ can be given by:
\begin{align}
    &f(v_3)=1, f(v_4)=f(v_5)=0 
    \implies  &g_{C_2}=0\nonumber\\
    &f(v_4)=1, f(v_3)=f(v_5)=0 
    \implies  &g_{C_2}=1\nonumber\\
    &f(v_5)=1, f(v_3)=f(v_4)=0 
    \implies  &g_{C_2}=2
\end{align}

Then relation $\mathcal{R}(\mathcal{G}^{(2,3)})$ induced by the clique labelling problem with tuples $(C_x,a,C_y,b)$ can be concretely given by:
\begin{align}
     &\mathcal{R}(\mathcal{G}^{(2,3)})=  
    \{(C_1,0,C_2,1),(C_1,0,C_2,2),(C_1,1,C_2,1),\nonumber\\
    &(C_1,1,C_2,2), (C_1,2,C_2,0),(C_2,1,C_1,0),(C_2,1,C_1,1),\nonumber\\&(C_2,2,C_1,0), (C_2,2,C_1,1),(C_2,0,C_1,2),(C_1,0,C_1,0),\nonumber\\&(C_1,1,C_1,1), (C_1,2,C_1,2),(C_2,0,C_2,0),(C_2,1,C_2,1),\nonumber\\&(C_2,2,C_2,2)\}  
\end{align}

 For this graph, the table of conditional probability $P(b|C_x,C_y,a)$ for all compatible labelling $a,b$ and maximum cliques $C_x,C_y$ is shown in table \ref{table:d3n2}.

\begin{table}[h]
\begin{center}
\begin{tabular}{|c c|ccc|ccc| } 
\hline
& & &$C_1$& & &$C_2$& \\
 & & $b=0$&$b=1$&$b=2$ & $b=0$&$b=1$&$b=2$\\
 \hline
     & $a=0$ & $1$ &$0$ & $0$& $0$& $*$ & $*$ \\
$C_1$&$a=1$  &$0$ & $1$ &$0$ & $0$& $*$ & $*$ \\
     &$a=2$  &$0$ & $0$ &$1$ & $1$& $0$ & $0$ \\
 \hline
     &$a=0$ & $0$& $0$ & $1$ & $1$&  $0$ &$0$  \\
$C_2$&$a=1$ & $*$& $*$ & $0$ &$0$ & $1$ &$0$  \\
     &$a=2$ & $*$& $*$ & $0$ &$0$ & $0$ &$1$\\
 \hline
\end{tabular}
\caption{Example of a table of conditional probabilities $P(b| C_x,C_y,a)$ corresponding to the distributed computation of $\mathcal{R}(\mathcal{G}^{(2,3)})$ (and relation reconstruction task) based on the graph in Fig. \ref{fig:consistentcolouring}. For distributed computation, {\it i.e.} satisfying (\textbf{T0}), the entries marked by $* \in [0,1]$ are free entries up to normalisation. For reconstruction of relation $\mathcal{R}(\mathcal{G}^{(2,3)})$, the free entries marked by $*$ belong to $(0,1)$ up to normalisation. For achieving optimal payoff $\mathcal{P}_{\mathcal{R}(\mathcal{G}^{(2,3)})}^{*}$,  all free elements marked by $*=0.5$.}
\label{table:d3n2}
\end{center}
\end{table}

 The entries marked with \textbf{*} are the free non-negative entries up to normalisation and the entries with $0$ or $1$ are constrained from the consistency conditions for the distributed clique labelling problem. This will give a table for which the conditions in ({\bfseries T0}) are satisfied. A table satisfying condition in ({\bfseries T1}) must have positive numbers at all the entries marked with \textbf{*}. In this example, a table satisfying condition ({\bfseries T2}) must have $0.5$ at all the entries marked with \textbf{*}.

\section{Proof of Theorem \ref{theo:D0classicalcomm}}\label{app:prooftheo1}
Before we delve into the proof, let us introduce a few notations that we will frequently use in this section. Prior to the distributed computation of the relation, Alice and Bob are given $\mathcal{G}^{(n,\omega)}$ and they construct a table $M$ whose entries are conditional probabilities $P(b| C_x,C_y,a)$ of compatible labels $a,b$, for all possible maximum cliques $C_x,C_y\in\mathcal{C}$. In this table the probability $p(b| C_x,C_y,a)\equiv((C_x,a,C_y,b))$ are entries corresponding to the event $(C_x,a,C_y,b)$ where $(C_x,a)\in X$, $C_y\in Y$ and $b\in B$.  The rows and the columns of this table are indexed as $(C_x,a)_\mathsf{r}$ and $(C_y,b)_\mathsf{c}$ respectively. In this table, the index runs over all the $a,b$ first and then updates the $C_x,C_y$. This table has $n\omega$ rows and $n\omega$ columns and may be perceived as a $n\times n$ block matrix with elements indexed $(C_x,C_y)$. We have $\mathbf{I}_{\omega\times \omega}$ on the diagonal blocks of the table as Bob has to output the same label as Alice whenever they get the same maximum cliques as input. The aforementioned distributed computation of the relation task can be mapped to the following properties (\textbf{T0}) of the table $M$. We have equivalence between the communication task and the table $M$ with the constraint (\textbf{T0}).
\begin{enumerate}[start=0,label={(\bfseries T\arabic*):}]
   \item Consistent labelling of cliques: If the event $(C_x,a,C_y,b)\notin \mathcal{R}(\mathcal{G}^{(n,\omega)})\implies P(b|C_x,a,C_y)=0$.
   \end{enumerate}
\begin{proof}
If Alice and Bob manage to compress the $nd$ rows of the table $M$ (i.e., the set of all possible inputs for Alice) into at least $\omega$ partitions such that no two rows in the same partition have entries in any columns that are different (may be due to constraints imposed by property (\textbf{T0}) or by choice filling the probabilities corresponding to events in $\mathcal{R}(\mathcal{G}^{(n,\omega)})$) then there exists a protocol that proves Theorem \ref{theo:D0classicalcomm}. Alice will communicate with Bob the partition to which her input belongs and then Bob can suitably pick a label for her input maximum clique $C_y$ while satisfying the probability distribution table that parties agreed upon at the start and thereby satisfying the consistency condition.

However, notice that there cannot be any less than $\omega$ number of partitions of the rows of the table $M$ satisfying (\textbf{T0}) such that no two rows in the same partition have entries in any columns that are different. This can be easily shown as every two rows corresponding to each block diagonal entry of $M$, \textit{i.e.} $(C_x=C_i,C_y=C_i)=\mathbf{I}_{\omega\times \omega}$, are distinct. Thus, each of the $\omega$ rows corresponding to Alice's input maximum clique $C_x=C_i$ must belong to a different partition.

This implies that every  disjoint partition $\tau(i) 
 (i\in \{0,1,\cdots,\omega-1\}$) of the rows described above must have exactly one row of the form $(C_x,a)$ for each maximum clique $C_x$, i.e., a row corresponding to exactly one out of all the possible label $a$ for every maximum clique $C_x$. In the following, we argue that there exists $\omega$  such disjoint partitions of rows. But before we proceed, we will list some properties of the table $M$ 
when such partitioning is possible.\\
 
 If there is an imposition that the rows of table $M$ can be partitioned into at least $\omega$ disjoint partitions $\tau(i)$ while satisfying the constraints discussed above this leads to restrictions on the structure of table $M$ that can be decided by both Alice and Bob in order to perform the distributed computation of relation specified by CLP.
 
 \begin{itemize}
     \item If some row (say $(C_x,a')_\mathsf{r}$) of off-diagonal block matrix $(C_x,C_y)$ has more than one non-zero entries (say($(C_x,a',C_y,\Tilde{b}))\neq0$ and $((C_x,a',C_y,\Tilde{b'}))\neq 0$) then the corresponding row in $M$ cannot belong to any partition that contains a row with index $(C_{x'(=y)},a)_\mathsf{r}$ where $a\in \{0,\cdots,\omega-1\}$ as there exist column $(C_{y},b)_\mathsf{c}$ where these two rows have different entries. This is because the block matrix $(C_y,C_y)=\mathbf{I}_{\omega\times \omega}$ and thus none of the rows have non-zero entries in two different columns in this block. Thus, this row must belong to a new partition thereby increasing the total number of partitions to $\omega+1$.
     \item If some column (say $(C_y,b')_\mathsf{c}$) of off-diagonal block matrix $(C_x,C_y)$ has more than one non-zero row entries then the rows corresponding to these nonzero entries in $M$ can only belong to the partition that contains the row $(C_{\Tilde{x}(=y)},\Tilde{a}(=b'))$. However, as discussed above exactly one out of all the possible labels $a$ for every maximum clique $C_x$ can belong to a partition. Therefore, Alice and Bob will be forced to create at least $\omega+1$ partitions. Therefore, if the number of partitions is restricted to $\omega$ then each row and column of every off-diagonal block matrix $(C_x, C_y)$ is some permutation $\Pi_{C_x,C_y}$ of $\mathbf{I}_{\omega\times\omega}$.
     \item The table must have the property $M=M^T$. If this does not hold then there exists an element for which $((C_x,a,C_y,b))=1\ne ((C_{x'(=y)},a'(=b),C_{y'(=x)},b'(=a)))$. The row $(C_x,a)_\mathsf{r}$ must belong to same partition as $(C_{x'(=y)},\Tilde{a}(=b))_\mathsf{r}$ as $((C_x,a,C_y,b))=1=((C_{x'(=y)},\Tilde{a},C_{y},b))$ only for $\Tilde{a}=b$. For any other allowed value of $\Tilde{a}$, $((C_{x'(=y)},\Tilde{a},C_{y},b))=0$. However, the row $(C_x,a)_\mathsf{r}$ and $(C_{x'(=y)},\Tilde{a}(=b))_\mathsf{r}$ have different entries in the column $(C_{y"(=x)},b"(=a))_\mathsf{c}$. $((C_x,a,C_{y"(=x)},b"(=a)))=1\neq ((C_{x'(=y)},a'(=b),C_{y'(=x)},b'(=a)))$. Thus, the row $(C_x,a)_\mathsf{r}$ cannot belong to any partition that contains a row indexed $(C_{x'(=y)}, \Tilde{a})$ where $\Tilde{a}\in\{0,\cdots,\omega-1\}$. 
 \end{itemize}

 Now, we will create a specific kind of $\omega$ disjoint partitions($\tau(i), i\in \{0,\cdots,\omega-1\}$) of the input received by Alice considering a probability table having the form discussed above.
 \begin{itemize}
 \item \textbf{Step 1:} $\forall a \in \{0,1,\cdots,\omega-1\}, (C_1,a)_\mathsf{r}\in\tau(a)$.
 \item \textbf{Step 2:} $\forall j\in\{2,\cdots,n\}$, say the block matrix $(C_1,C_j)$ is a permutation matrix $\Pi_{1,C_j}$ then the row $(C_j,a')_\mathsf{r}\in \tau(a)$ where $a'$ is the $(a)^{th}$ element of $\Pi_{1,C_j}*(0~1~\cdots~\omega-1)^T$.

\end{itemize}
When Alice communicates the partition to which her input $(C_x,a)$ belongs, Bob can pick the label for maximum clique $C_y$ that obeys the consistency labelling of pairwise clique condition. It is important to note that each row associated with Alice's input maximum clique $C_x$ must belong to a distinct partition else Bob might not be able to assign a label obeying the consistency condition.

For example, consider the graph shown in Fig. \ref{fig:consistentcolouring}. Alice and Bob adopt a deterministic strategy and fill the free entries marked with \textbf{*} in Table \ref{table:d3n2} with $0$s and $1$s as seen in Table \ref{table:d3n2_T0}.

\begin{table}[h]
\begin{center}
\begin{tabular}{|c c|ccc|ccc| } 
\hline
& & &$C_1$& & &$C_2$& \\
 & & $b=0$&$b=1$&$b=2$ & $b=0$&$b=1$&$b=2$\\
 \hline
     & $a=0$ & $1$ &$0$ & $0$ & $0$& $0$ & $1$ \\
$C_1$&$a=1$  &$0$ & $1$ &$0$ & $0$& $1$ & $0$ \\
     &$a=2$  &$0$ & $0$ &$1$ & $1$& $0$ & $0$ \\
 \hline
     &$a=0$ & $0$& $0$ & $1$ & $1$& $0$ &$0$  \\
$C_2$&$a=1$ & $0$& $1$ & $0$ &$0$ & $1$ &$0$  \\
     &$a=2$ & $1$& $0$ & $0$ &$0$ & $0$ &$1$\\
 \hline
\end{tabular}
\caption{Example of a table of conditional probabilities $P(b| C_x,C_y,a)$ for the graph in Fig. \ref{fig:consistentcolouring} satisfying (\textbf{T0}).}
\label{table:d3n2_T0}
\end{center}
\end{table}
For Table \ref{table:d3n2_T0}, we can make three partitions $\tau(0),\tau(1)$ and $\tau(2)$ for the rows such that exactly one row of each maximum clique belongs to a partition. In this the partitions are $\tau(0)=\{(C_1,a=0)_\mathsf{r},(C_2,a=2)_\mathsf{r}\}$, $\tau(1)=\{(C_1,a=1)_\mathsf{r},(C_2,a=1)_\mathsf{r}\}$ and $\tau(2)=\{(C_1,a=2)_\mathsf{r},(C_2,a=0)_\mathsf{r}\}$. Upon receiving $C_x$ and $a$ in each round Alice can send $i$ corresponding to $\tau(i)$. After knowing the partition $\tau(i)$, Bob can always pick the label for his maximum clique $C_y$ that does not violate the consistency condition. Thus, a classical three-level system is sufficient for performing the distributed computation of $\mathcal{R}(\mathcal{G}^{(2,3)})$.
\end{proof}

\section{Proof of Lemma \ref{theo:D1classicalcomm}}\label{app:prooflemma1}

\begin{proof}
Before the game begins, Alice and Bob construct the table $M$ of conditional probabilities $P(b| C_x,C_y,a)$ which has $n\omega$ rows and $n\omega$ columns. We will refer to each of the rows and columns of this table as $(C_x,a)_\mathsf{r}$ and $(C_y,b)_\mathsf{c}$ respectively. If $M$ satisfies  {\bfseries T(0)-T(1)} then upon receiving $(C_x,a)$ Alice communicates the relevant row to Bob and relation reconstruction becomes possible as Bob with $P(b| C_x,C_y,a)$ outputs clique label $b$ for his input maximum clique $C_y$. Therefore, we have a trivial upper bound on the dimension of the classical system which is required as $n\omega$. The deterministic classical strategy employed for Theorem \ref{theo:D0classicalcomm} cannot reconstruct the relation since, the conditional probability table must contain nonzero entries corresponding to all events $(C_x,a,C_y,b)\in \mathcal{R}(\mathcal{G}^{(n,\omega)})$. 

Therefore, we cannot use the strategy as used before. Nonetheless, observe that if two rows of the probability table can be made identical while satisfying the consistency condition, Alice and Bob can encode them in the same communication message. In the table, there is redundancy when the same vertex shows up in different maximum cliques. For instance, let vertex $v$  be in both maximum clique $C_i$ and $C_j$. Then the rows in the conditional probability table corresponding to $(C_x=C_i,a)_\mathsf{r}$ and $(C_{x}=C_j,a')_\mathsf{r}$, where label $a$ and $a'$ for the maximum clique $C_i$ and $C_{j}$ respectively colour the vertex $v$ as $1$, can be assigned the same entries. For such a vertex $v$, $(C_x=C_i,a,C_y=C_j,b=a'),(C_x=C_j,a',C_y=C_i,b=a),(C_x=C_j,a',C_y=C_j,b=a'),(C_x=C_i,a,C_y=C_i,b=a)\in\mathcal{R}(\mathcal{G}^{(n,\omega)})$. Also for any other $C_y (\ne C_i,C_j)$, $(C_x=C_i,a,C_y,b)\in \mathcal{R}(\mathcal{G}^{(n,\omega)})\implies(C_x=C_j,a',C_y,b)\in\mathcal{R}(\mathcal{G}^{(n,\omega)})$. Thus, the entries in the table $M$ corresponding to the rows $(C_x=C_i,a)_\mathsf{r}$ and $(C_{x}=C_j,a')_\mathsf{r}$ can be assigned identically (especially the non-zero entries) while guaranteeing relation reconstruction without violation of the consistency condition. The entries that are necessarily zero in one of the rows are also zero for the other row. Therefore, these two inputs $(C_x=C_i,a)_\mathsf{r}$ and $(C_{x}=C_j,a')_\mathsf{r}$ can be encoded in the same message. 

Therefore, Alice and Bob can remove all redundant rows in this manner and end up with an encoding based on the compressed table which now has $|\mathcal{V}|$ distinct rows, where each row corresponds to a distinct vertex. Sufficiency of communicating $|\mathcal{V}|$-level classical system follows trivially since it allows Alice to send all information about her input. $|\mathcal{V}|\leq n\omega$ is saturated if all the maximum cliques are disconnected in the given graph. 

Now we prove the necessity of  $|\mathcal{V}|$-level classical system 
 to  perform the relation reconstruction
 when considering $\mathcal{R}(\mathcal{G}^{(n,\omega)})$. For every vertex $v$ in a maximum clique $C_i$ where $i \in \{1,\cdots,n\}$ there is an input corresponding to this vertex for Alice $(C_i,a)$ where label $a$ assigns colour $1$ to $v$ and rest of the vertices in the maximum clique are assigned $0$. This is due to condition {(\bfseries G0)}. For any maximum clique $C_i$, each of the Alice's input, $(C_i,a)$ where $a\in \{0,\cdots,\omega-1\}$, must be encoded with different message alphabet. This is because Bob needs to exactly guess the input clique label of Alice whenever his input is $C_y=C_i$. Now for any two vertices $v,v'$ that belong to two different maximum cliques, say $v\in C_i$ and $v' \in C_j$ (where $i \neq j$), there exists a maximum clique $C_k$ and a vertex $u$ $\in C_k$ (where $k$ maybe $i$, $j$ or some other number) such that it is orthogonal exactly to one of these vertices (say $v$ w.l.o.g.). This is due to condition {(\bfseries G1)}. Let Alice's input corresponding to $v$ and $v'$ be $(C_x=C_i,a)$ and $(C_x=C_j,a')$ respectively and Bob has input $C_y=C_k$. For these rounds, $P(b|C_x=C_i,a,C_y=C_k)=0$ and $P(b|C_x=C_j,a',C_y=C_k)> 0$ where Bob's output label $b$ for maximum clique $C_y=C_k$ assigns $1$ exactly to vertex $u$. Thus, the inputs corresponding to every pair of vertices that do not belong to the same maximum clique must be encoded with different message alphabets to accomplish relation reconstruction and obtain a non-zero payoff. Since there are $|\mathcal{V}|$ vertices, the classical message must be encoded in a system of dimension $\geq |\mathcal{V}|$.

 We now argue that locally randomising over the deterministic encoding and decoding protocols (that is the usage of {\it Private Coins}) cannot lower the necessary classical communication, that is, using less than $\log_2 |\mathcal{V}|$ bit along with private coins cannot accomplish ({\bfseries T0})-({\bfseries T1}). To see this, we will consider a convex combination of deterministic encoding of Alice for protocols with communication of a $(|\mathcal{V}|-1)$-level classical system. Consider some maximum clique $C_i$. In any deterministic encoding, each of Alice's input $(C_i,a)$ where $a\in \{0,\cdots,\omega-1\}$, must be encoded with a different message alphabet. This is so because if Alice and Bob receive the same maximum clique as input, then their labelling for the maximum clique must match. Also, this deterministic encoding will encode some of the inputs corresponding to different vertices, say $v$ and $v'$, in the same message alphabet. There will be some inputs  $(C_x=C_i,a)$ and $(C_x=C_{j(\neq i)},a')$ encoded in the same message. Here labels $a$ and $a'$ assign binary colour $1$ to $v$ in $C_i$ and $v'(\neq v)$ in $C_j$ respectively while the rest of the vertices in these maximum cliques are assigned $0$. Individually, each of these encodings will be unsuccessful in relation reconstruction. Furthermore, since Alice and Bob do not have access to public coin, therefore, Bob is not aware of Alice's choice of encoding in a given round. Thus, Bob cannot use decoding that is correlated with Alice's encoding strategy in a given round. If Bob tries to satisfy consistency conditions then he will not be able to have non-zero probability corresponding to all the events $(C_x,a,C_y,b)\in\mathcal{R}(\mathcal{G}^{(n,\omega)})$.\\

 Thus, communication of $|\mathcal{V}|$-level classical system is necessary for reconstruction of the relation $\mathcal{R}(\mathcal{G}^{(n,\omega)})$.        
\end{proof}

\section{Proof of Theorem \ref{theo:paleyg2}}\label{app:prooftheo4}
\begin{proof}
We are looking to compute the dimension of the faithful representation that gives the optimal solution to the Lov\'{a}sz-theta optimisation of $\mathcal{G}_{Paley(q)}$, i.e., we want to find the minimum $\xi^*(\mathcal{G}_{Paley(q)})$ such that $|u_i \rangle \in \mathbb{R}^{\xi^*(\mathcal{G}_{Paley(q)})}$ for the vectors $|u_i \rangle \in S_{opt}$. This quantity is given by the rank of the Gram Matrix $M_{opt}$ of the set of (normalised) vectors $S_{opt}$. We have that
\[(M_{opt})_{k,l} = \begin{cases} 
      1 & k = l \\
      0 & k \sim l \\
      2/(q^{1/2}+1) & (k \nsim l) \wedge (k \neq l). 
   \end{cases}
\]
In other words, $M_{opt} = I + \frac{2}{q^{1/2} + 1} A(\mathcal{\overline{G}}_{Paley(q)})$, where $A(\mathcal{\overline{G}}_{Paley(q)})$ denotes the adjacency matrix of the complement of $\mathcal{G}_{Paley(q)}$ (which is isomorphic to $\mathcal{G}_{Paley(q)}$ since the graph is self-complementary). 

To compute $rank(M_{opt})$, we calculate its spectrum and show that it has exactly $(q+1)/2$ non-zero eigenvalues, so that $rank(M_{opt}) = (q+1)/2$. 

To do this, we compute the spectrum of $A(\mathcal{\overline{G}}_{Paley(q)}) = A(\mathcal{G}_{Paley(q)})$. Following \cite{Lovasz07}, let us define a matrix $K$ based on the quadratic characters $\chi(k-l)$ 
\begin{align}
     \chi(k-l) = \begin{cases}
		1 & (k-l) \; \text{is quadratic residue modulo} \; q \\
		0 & k = l \\
		-1 & \text{else}.
	\end{cases}
\end{align}
 by $K_{k,l} = \chi(k-l)$. By the property of the characters that $\chi(xy) = \chi(x) \chi(y)$ and $\sum_{x=0}^{q-1} \chi(x) = 0$, we can prove the result $K^2 = q I - J$ where $J$ denotes the all-ones matrix (See Lemma \ref{lemma:paley} in Appendix \ref{app:lemmaforpaley} for the proof).

We also see by direct term-by-term comparison that the adjacency matrix of the Paley graph can be written as
\begin{equation}
A(\mathcal{G}_{Paley(q)}) = \frac{1}{2} \left(K + J - I \right).
\end{equation} 
We, therefore, obtain that
\begin{equation}
\left( A(\mathcal{G}_{Paley(q)}) \right)^2 = \frac{q-1}{4} \left(J + I \right) - A(\mathcal{G}_{Paley(q)}).
\end{equation}
Now observe that the all-ones vector $|j \rangle$ is an eigenvector of $A(\mathcal{G}_{Paley(q)})$ and consider another eigenvector $| e_{\lambda} \rangle$ corresponding to eigenvalue $\lambda \neq 0$. Since $|e_{\lambda} \rangle$ is orthogonal to $|j \rangle$, we have that $J | e_{\lambda} \rangle = 0$, so that
\begin{eqnarray}
\left( A(\mathcal{G}_{Paley(q)}) \right)^2 | e_{\lambda} \rangle = \lambda^2 | e_{\lambda} \rangle = \left(\frac{q - 1}{4} - \lambda \right) | e_{\lambda} \rangle, 
\end{eqnarray}
or in other words that
\begin{eqnarray}
&&\lambda^2 + \lambda - \frac{q-1}{4} = 0, \nonumber \\
&& \implies \lambda = \frac{1}{2} \left(-1 \pm q^{1/2} \right).
\end{eqnarray}
Thus, the spectrum and corresponding degeneracies of $A\left(\mathcal{\overline{G}}_{Paley(q)} \right)$ are found to be
\[ \text{spec}\left(A\left(\mathcal{\overline{G}}_{Paley(q)} \right) \right) = \begin{cases}
		(q-1)/2 & \quad 1 \\
		\frac{1}{2} \left(-1 + q^{1/2} \right) & \quad (q-1)/2 \\
		\frac{1}{2} \left(-1 - q^{1/2} \right) & \quad (q-1)/2.
	\end{cases}
\]

As we have seen, the Gram Matrix $M_{opt}$ from the optimal representation giving rise to $\theta(\mathcal{G}_{Paley(q)})$ is given by
\begin{equation}
M_{opt} = I + \frac{2}{q^{1/2} + 1} A\left(\mathcal{\overline{G}}_{Paley(q)} \right).
\end{equation}
Therefore, the spectrum of $M_{opt}$ consists of exactly $(q+1)/2$ non-zero eigenvalues given by
\[ \text{spec}\left(M_{opt} \right) = \begin{cases}
		\sqrt{q} & \quad 1 \\
		2\sqrt{q}/(1+\sqrt{q}) & \quad (q-1)/2 \\
		0 & \quad (q-1)/2.
	\end{cases}
\]
Therefore, we obtain that $rank(M_{opt}) = \xi^*(\mathcal{G}_{Paley(q)}) = (q+1)/2$.

We can even go further and note that since the adjacency matrix of the Paley graph is a circulant matrix (the $k$-th row is a cyclic permutation of the $1$-st row with offset $k$), the \textit{eigenvectors} of the adjacency matrix $A(\mathcal{\overline{G}}_{Paley(q)})$ (and therefore the Gram Matrix $M_{opt}$) are the Fourier vectors
\begin{equation}
|e_{\lambda} \rangle = \frac{1}{q} \left(1, \omega^{\lambda}, \omega^{2 \lambda}, \ldots, \omega^{(q-1) \lambda} \right),
\end{equation} 
with $\lambda = 0, 1, \ldots, q-1$, where $\omega = \exp\left(\frac{2 \pi i}{q} \right)$ is a primitive $q$-th root of unity. Note that $|e_0 \rangle = |j \rangle$ is the all-ones vector. We can then explicitly calculate that

\begin{equation}
\begin{split}
M_{opt} & |e_{\lambda} \rangle = \\ &\left[ \frac{\sqrt{q} -1}{\sqrt{q} + 1} + \frac{1}{\sqrt{q} + 1} \sum_{\substack{l: l \neq 1 \\ (1-l) \; \text{is a quad. res. mod}~ q}} \omega^{(l-1) \lambda} \right.\\
& \left. - \frac{1}{\sqrt{q} + 1} \sum_{\substack{l: l \neq 1 \\ (1-l) \; \text{is not a quad. res. mod}~ q}} \omega^{(l-1) \lambda} \right]  |e_{\lambda} \rangle
\end{split}
\end{equation}
We can then explicitly compute for prime $q$ not only the eigenvalues of $M_{opt}$ as above but also see that the eigenvalue $\sqrt{q}$ corresponds to the eigenvector $|j \rangle = |e_0 \rangle$, the eigenvalues $2 \sqrt{q}/(1 + \sqrt{q})$ correspond to the eigenvectors $|e_{\lambda} \rangle$ for $\lambda$ being the remaining quadratic residues modulo $q$, and the zero eigenvalues correspond to the eigenvectors $|e_{\lambda} \rangle$ for $\lambda$ being the quadratic non-residues modulo $q$. 

\end{proof}

\section{Proof of Result used in Theorem \ref{theo:paleyg2}}\label{app:lemmaforpaley}

Consider a matrix $K$ (as defined in theorem \ref{theo:paleyg2}) based on the quadratic characters $\chi(k-l)$
\begin{align}
     \chi(k-l) = \begin{cases}
		1 & (k-l) \; \text{is quadratic residue modulo} \; q \\
		0 & k = l \\
		-1 & \text{else}.
	\end{cases}
\end{align}

by $K_{k,l} = \chi(k-l)$. Using the property of the characters that $\chi(xy) = \chi(x) \chi(y)$ and $\sum_{x=0}^{q-1} \chi(x) = 0$, we will now prove the result that we have used in theorem \ref{theo:paleyg2}.
\begin{lemma}
\label{lemma:paley}
$K^2 = q I - J$, where $J$ denotes the all-ones matrix.
\end{lemma}
\begin{proof}
We want to prove that the diagonal entries of $K^2$ are equal to $(q-1)$ and the off-diagonal entries are equal to $-1$. The diagonal entries are given by the squared norms of the columns of $K$, which have one zero entry, $(q-1)/2$ entries of value $1$ (corresponding to the quadratic residues modulo $q$ and the degree of each vertex in $\mathcal{G}_{Paley(q)}$) and $(q-1)/2$ entries of value $-1$. Therefore, the squared norms of the columns and hence the diagonal entries of $K^2$ are equal to $q-1$.

The off-diagonal entries $(K^2)_{k,l}$ are given by $\sum_{j=0}^{q-1} \chi(k-j) \chi(l-j) = \sum_{j'=0}^{q-1} \chi(j') \chi((l-k)+j')$. Since $\chi(0) = 0$, the term for $j'=0$ vanishes and we have $\sum_{j'=1}^{q-1} \chi(j') \chi((l-k)+j')$. Since $\chi(j') \in \{\pm 1\}$ for $j' \neq 0$, the sum reduces to $\sum_{j'=1}^{q-1}  \chi((l-k)+j')/\chi(j') = \sum_{j'=1}^{q-1} \chi((l-k)/j' + 1)$ where we used the property of the characters that $\chi(xy) = \chi(x) \chi(y)$. We finally obtain $\sum_{j'=1}^{q-1} \chi((l-k)/j' + 1) = \left[\sum_{j'' = 0}^{q-1} \chi(j'') \right] - \chi(1) = 0 - 1 = -1$ where we used the property that as $j'$ ranges over $[q-1]$, the argument $(l-k)/j' + 1$ ranges over elements $\{0,\ldots, q-1\} \setminus \{1\}$. Therefore, we obtain the off-diagonal entries to be $-1$ thus showing that $K^2 = q I - J$.
\end{proof}

\section{Proof of Theorem \ref{theo:dntod2nlowerbound}}\label{app:prooftheo5}
\begin{proof} 
The amount of classical public coin, while communicating $\omega$-level classical system, depends on the graph and can be upper bounded by the total number of different classical deterministic encoding and decoding strategies (or the total number of different tables of conditional probabilities that Alice and Bob can prepare while satisfying the constraints mentioned in Appendix \ref{app:prooftheo1}). We observe that among different graphs $\mathcal{G}$ with the same number of maximum cliques $n$ of clique size $\omega$, the graph in which all maximum cliques are disconnected requires the most amount of classical public coin assistance. On the other hand, graphs in which every maximum clique shares the most number of its vertices with other maximum cliques, require the least amount of classical public coin due to the least number of $*$ entries in their conditional probability table (for example in Table \ref{table:SRdto2}). We also know that the most number of vertices that any two maximum cliques can share is $\omega-2$ to have an orthogonal representation in $\C^{\omega}$ (Proposition \ref{prop:lovasz}). An example is provided in Fig. \ref{fig:graphlowerbound} for $\omega=5$ and $n=2$.

\begin{figure}[h]
     \centering
\begin{DiagramV}[1.5]{0}{0}
\begin{move}{0,0}
\fill[black] (0.5,0) circle (0.25);
\fill[black] (0,-2) circle (0.25);
\fill[black] (0,2) circle (0.25);
\fill[black] (-2,1) circle (0.25);
\fill[black] (-2,-1) circle (0.25);
\fill[black] (2,1) circle (0.25);
\fill[black] (2,-1) circle (0.25);
\draw (0,2) -- (0,-2) -- (2,-1) -- (2,1) -- (0,2)--(2,-1)--(0.5,0)--(2,1)--(0,-2);
\draw (0,2) -- (0,-2) -- (-2,-1) -- (-2,1) -- (0,2) --(-2,-1)--(0.5,0)--(-2,1)--(0,-2)--(0.5,0)--(0,2);
\draw (0,-4) node {\large $\mathcal{G}_1$};
\end{move}
\end{DiagramV}
\begin{DiagramV}[1.5]{-2}{0}
\begin{move}{0,0}
\fill[black] (-1.5,0) circle (0.25);
\fill[black] (0,-2) circle (0.25);
\fill[black] (0,2) circle (0.25);
\fill[black] (2,1) circle (0.25);
\fill[black] (2,-1) circle (0.25);
\draw (0,2) -- (0,-2) -- (2,-1) -- (2,1) -- (0,2)--(2,-1)--(-1.5,0)--(2,1)--(0,-2) --(-1.5,0) --(0,2);
\draw (3,-4) node {\large $\mathcal{G}_2$};
\end{move}
\end{DiagramV}
\begin{DiagramV}[1.5]{0}{0}
\begin{move}{0,0}
\fill[black] (1.5,0) circle (0.25);
\fill[black] (0,-2) circle (0.25);
\fill[black] (0,2) circle (0.25);
\fill[black] (-2,1) circle (0.25);
\fill[black] (-2,-1) circle (0.25);
\draw (0,2) -- (0,-2) -- (-2,-1) -- (-2,1) -- (0,2)--(-2,-1)--(1.5,0)--(-2,1)--(0,-2) --(1.5,0) --(0,2);
\end{move}
\end{DiagramV}
\caption{Two Graphs with ($\omega=5, n=2$), the graph on the left $\mathcal{G}_1$ has two maximum cliques of size $\omega=5$ and $\omega-2=3$ vertices common between these maximum cliques. The graph on the right $\mathcal{G}_2$ consists of two disconnected maximum cliques of size $\omega=5$.}
    \label{fig:graphlowerbound}
\end{figure}
To find the lower bound on the classical public coin for a graph with $n$ maximum cliques of size $\omega$, we can calculate the amount of classical public coin necessary for a graph where all the maximum size cliques share $\omega-2$ vertices. Such a graph will saturate the lower bound.  The order of the graph in this case is $|\mathcal{V}|=\omega + 2(n-1)$. 

We also observe that for such a graph the associated conditional probability $n\omega \times n\omega$ table with entries $P(b|C_x,C_y,a)$, the number and structure of the free entries $*$ is equivalent to that of a graph with $n$ disconnected maximum cliques of size $\omega=2$. Thus, the number of classical deterministic strategies and therefore classical public coin required for these two graphs is the same. For example, in the case of the graph shown in Fig. \ref{fig:consistentcolouring} (or left side of Fig. Fig. \ref{fig:SRdto2}) we see that Table \ref{table:SRdto2} is the conditional probability table which is also equivalent (in terms of $*$) to the conditional probability table for the graph on the right $\mathcal{G}^{(n=2,\omega=2)}$ in Fig. \ref{fig:SRdto2}. 
     \begin{figure}[h]
    \centering
\begin{DiagramV}{0}{0}
\begin{move}{0,0} 
\fill[black] (0,0) circle (0.25);
\draw (0,0-1) node {$v_1$};
\fill[black] (6,0) circle (0.25);
\draw (6,0-1) node {$v_3$};
\fill[black] (3,5.196) circle (0.25);
\draw (3,5.196+1) node {$v_2$};
\draw (0,0) -- (6,0);
\draw (0,0) -- (3,5.196);
\draw (3,5.196) -- (6,0);
\draw (3,5.196/2 -0.5) node{$C_1$};
\fill[black] (3+6,5.196) circle (0.25);
\draw (3+6,5.196+1) node {$v_4$};
\fill[black] (6+6,0) circle (0.25);
\draw (6+6,0-1) node {$v_5$};
\draw (0+6,0) -- (6+6,0);
\draw (0+6,0) -- (3+6,5.196);
\draw (3+6,5.196) -- (6+6,0);
\draw (3+6,5.196/2 -0.5) node {$C_2$};
\draw (14,5.196/2 -0.5) node {\large $\cong$};
\end{move}
\end{DiagramV}
\begin{DiagramV}{4}{0}
\begin{move}{0,0} 
\fill[black] (0,0) circle (0.25);
\draw (0,0-1) node {$v_1$};
\fill[black] (3,5.196) circle (0.25);
\draw (3,5.196+1) node {$v_2$};
\draw (0,0) -- (3,5.196);
\draw (3,5.196/2 -0.5) node { $C_1$};
%
\fill[black] (3+6,5.196) circle (0.25);
\draw (3+6,5.196+1) node {$v_4$};
\fill[black] (6+6,0) circle (0.25);
\draw (6+6,0-1) node {$v_5$};
\draw (3+6,5.196) -- (6+6,0);
\draw (3+6,5.196/2 -0.5) node { $C_2$};
\end{move}
\end{DiagramV}   
    \caption{Calculating the lower bound on classical public coin for a graph with $\omega$ sized $n$ maximum cliques with $\omega-2$ common vertices is equivalent (in terms of the number of classical deterministic strategies) to a graph with $\omega=2$ sized $n$ disconnected maximum cliques, as shown in this example continuing the simple example provided before in Fig \ref{fig:consistentcolouring}.}
    \label{fig:SRdto2}
\end{figure}
\begin{table}[h]
\begin{center}
\begin{tabular}{|c c|ccc|ccc| } 
\hline
& & &$C_1$& & &$C_2$& \\
 & & $b=0$&$b=1$&$b=2$ & $b=0$&$b=1$&$b=2$\\
 \hline
     & $a=0$ & $1$ &$0$ & $0$& $0$& $*$ & $*$ \\
$C_1$&$a=1$  &$0$ & $1$ &$0$ & $0$& $*$ & $*$ \\
     &$a=2$  &$0$ & $0$ &$1$ & $1$& $0$ & $0$ \\
 \hline
     &$a=0$ & $0$& $0$ & $1$ & $1$& 0 &$0$  \\
$C_2$&$a=1$ & $*$& $*$ & $0$ &$0$ & $1$ &$0$  \\
     &$a=2$ & $*$& $*$ & $0$ &$0$ & $0$ &$1$\\
 \hline
\end{tabular}
~~$\cong$~~
\begin{tabular}{|c c|cc|cc| } 
\hline
& & $C_1$& & $C_2$& \\
 & & $b=0$&$b=1$ & $b=0$&$b=1$\\
 \hline
$C_1$   & $a=0$ & $1$ &$0$ & $*$ & $*$ \\
        &$a=1$  &$0$ & $1$& $*$ & $*$ \\
 \hline
$C_2$   &$a=0$ &$*$& $*$ & $1$ &$0$  \\
        &$a=1$ & $*$& $*$  &$0$ & $1$  \\
 \hline
\end{tabular}
\caption{The conditional probability table version of the equivalence based on the two graphs in Fig. \ref{fig:SRdto2}, show that the Table for $\omega=3$ sized $n=2$ maximum cliques with $\omega-2=1$ common vertices is equivalent (in terms of number of classical deterministic strategies) to Table with $\omega=2$ sized $n=2$ disconnected maximum cliques.}
\label{table:SRdto2}
\end{center}
\end{table}
 Therefore we have shown that we can calculate the lower bound for the classical public coin required for a graph with $\omega$-sized $n$ maximum cliques by calculating the classical public coin needed for a graph with $n$ maximum cliques of size $2$ that are disconnected.

\end{proof}

\section{Proof of Corollary \ref{corollary:necessary SR}}\label{app:proofcorollaary1}
\begin{proof}
 Using Theorem \ref{theo:dntod2nlowerbound}, to find the lower bound on classical public coin required for a graph $\mathcal{G}^{(n,\omega)}$ to accomplish reconstruction of relation $\mathcal{R}(\mathcal{G}^{(n,\omega)})$, {\it i.e.} satisfy ({\bfseries T0}) and ({\bfseries T1}), we equivalently calculate the classical public coin required for the same task when considering a graph $\mathcal{G}^{(n,\omega=2)}$ where any two maximum cliques are disconnected and one bit communication is allowed.

 For reconstruction of $\mathcal{R}(\mathcal{G}^{(n,\omega=2)})$, we require a convex combination of deterministic strategies while communicating a classical system of $\omega$ dimension. The conditional probability table $M$ resulting from the convex combination of these strategies must have positive entries in the off-diagonal block matrix $(C_x,C_{y\neq x})$ while the diagonal block matrices $(C_x,C_{x})$ must be equal to the identity matrix $\mathbf{I}_{2}$. Thus, in the conditional probability table $M$, for every $(C_x,a,C_{y\neq x})$ there must be at least a pair of deterministic table/ strategy such that one has $P(b|C_x,a,C_y)=0$ and the other table has $P(b|C_x,a,C_y)=1$ as its entry. Any classical deterministic strategy constitutes of filling the table of conditional probability such that every off-diagonal block matrix of this table $(C_x,C_{y\neq x})$ is either  $\mathbf{I}_{2}$ or $\sigma_x$, where $\sigma_x$ is the Pauli-$x$ operator (or the NOT operator). The set of $n$ classical strategies to achieve reconstruction are the following. The $i^{th}$ strategy corresponds to the table where only off-diagonal block matrix $(C_1,C_i)=\sigma_x$ and rest $(C_1,C_{j(\ne i)})=\mathbf{I}_2~ \forall i\in \{2,\cdots,n\}$. Note that fixing the block matrices in the first row alone fixes the entire table if the amount of classical communication is restricted to $1$ bit (See Appendix \ref{app:prooftheo1}).

Note that taking each such $n$ deterministic classical strategy discussed earlier and their convex combinations yields a table of conditional probabilities $P(b|C_x,a,C_y)$, $M$, that leads to some non-zero payoff. It is worth mentioning that the payoff for the above strategy is $\mathcal{P}_{\mathcal{R}(\mathcal{G}^{(n,2)})}=\frac{1}{n} >0$ and since a non-zero payoff ensures relation reconstruction, thus we satisfy ({\bfseries T0}) and ({\bfseries T1}).  However, this does not always help achieve optimal payoff $\mathcal{P}_{\mathcal{R}(\mathcal{G}^{(n,2}))}^{*}$ for the graph under consideration.\\
\end{proof}

\section{Proof of Corollary \ref{cor:SR2}}\label{app:proofcorollaary2}
\begin{proof}
 Consider a graph $\mathcal{G}^{(n,\omega)}$ satisfying {(\bfseries G0)}-{(\bfseries G2)} with faithful orthogonal range $\omega$. From Theorem \ref{theo:dntod2nlowerbound}, the lower bound on classical public coin assistance to $\log_2\omega$ bit channel required for relation reconstruction and obtaining optimal payoff ({\bfseries T2}) is equal to the lower bound on classical public coin required to achieve $\mathcal{P}_{\mathcal{R}(\mathcal{G}^{(n,2)})}^{*}=0.5$ for a graph $\mathcal{G}^{(n,\omega=2)}$ with $n$ disconnected maximum size cliques while communicating one bit.

Similar to the proof of Corollary \ref{corollary:necessary SR} we will again consider a convex combination of deterministic strategies while communicating a classical system of dimension two in the latter task. For any such deterministic strategy, the associated conditional probability table has every off-diagonal block  $(C_x, C_y)$ to be either $\mathbf{I}_2$ or $\sigma_x$, where $\mathbf{I}_2$ is $2\times2$ identity matrix. Also, for any such deterministic strategy $(C_x,C_y)=(C_1,C_x)\oplus_2(C_1,C_y)$ where $\mathbf{I}_2\rightarrow0$ and $\sigma_x\rightarrow 1$. Note that fixing the block matrices in the first row alone fixes the entire table if the amount of classical communication is restricted to $1$ bit (See Appendix \ref{app:prooftheo1}).

In the final table of conditional probability $M$, we want each entry in every off-diagonal block $(C_x,C_y)$ to be $0.5$. This is possible if we have a uniform convex mixture of deterministic tables where half of them have $(C_x,C_y)=\sigma_x$ and the rest have $(C_x,C_y)=\mathbf{I}_2$ such that the effective weight for each free entry $*$ is $0.5$.
For $n=2$, convex combination of two deterministic tables, one with $(C_1,C_2)=\mathbf{I}_2$ and other with $(C_1,C_2)=\sigma_x$, give payoff $\mathcal{P}_{\mathcal{R}(\mathcal{G}^{(2,2)})}=0.5$. For $n=3$, we need four tables i.e. $(C_1,C_2)=\mathbf{I}_2$ or $\sigma_x$
and $(C_1,C_3)=\mathbf{I}_2$ or $\sigma_x$. All convex combinations of any subset of these four deterministic strategies/tables will lead to a table $M$ where some of the off-diagonal block matrix $(C_x,C_y)$ have a contribution from an unequal number of $\mathbf{I}_2$ and $\sigma_x$. For $n\geq 2$, by the similar argument we need a minimal collection of deterministic tables such that corresponding to every two-block matrix of the form $(C_1,C_{j\neq 1})$ and $(C_1,C_{j'\neq 1})$, there are an equal number of tables where $(C_1,C_{j})=\mathbf{I}_2$ and $(C_1,C_{j'})=\mathbf{I}_2$, $(C_1,C_{j})=\mathbf{I}_2$ and $(C_1,C_{j'})=\sigma_x$, $(C_1,C_{j})=\sigma_x$ and $(C_1,C_{j'})=\mathbf{I}_2$, and 
$(C_1,C_{j})=\sigma_x$ and $(C_1,C_{j'})=\sigma_x$. This is exactly the problem for orthogonal arrays that have been discussed above if we substitute $\mathbf{I}_2\rightarrow 0$ and $\sigma_x\rightarrow 1$. In other words, this corresponds to the minimum number of rows required so that any pair of columns have in their rows the entries $\{(0,0),(0,1),(1,0),(1,1)\}$ occurring an equal number of times.
Thus, for a graph $\mathcal{G}^{(n,\omega=2)}$ with $n(>2)$ the parties Alice and Bob need classical public coin with $T_{n-1}$ input to get optimal payoff $\mathcal{P}_{\mathcal{R}(\mathcal{G}^{(n,2)})}^{*}=0.5$ when they are allowed to communicate $\omega$ level classical system. For any graph $\mathcal{G}^{(n,\omega>2)}$ considered here, Alice and Bob need classical public coin with at least $T_{n-1}$ input to get optimal payoff $\mathcal{P}_{\mathcal{R}(\mathcal{G}^{(\omega)})}^{*}$ for the relation $\mathcal{R}(\mathcal{G}^{(n,\omega>2)})$ when they are allowed to communicate $\omega$ level classical system. This completes the proof.  
\end{proof}

\bibliographystyle{ieeetr}

\end{document}